\documentclass{llncs}
\usepackage{cite}

\usepackage{latexsym,amsmath,amssymb,url}
\usepackage{xspace}

\usepackage{tikz}
\usetikzlibrary{snakes}
\usetikzlibrary{arrows}
\usepackage{caption}
\usepackage{subcaption}
\usepackage{subfig}

\newcommand{\pw}{\textrm{pw}}

\newcommand{\PBS}{{\sc PBS}\xspace}
\newcommand{\MCPP}{{\sc MCPP}\xspace}
\newcommand{\compMCPP}{{\sc comp-MCPP}\xspace}

\newcommand{\remove}[1]{}
\newcommand{\add}[1]{#1}

     \tikzstyle{city}=[circle,draw=black,fill=black, minimum size=3mm]
\title{Structural Parameterizations of the Mixed Chinese Postman Problem}
\author{Gregory Gutin, Mark Jones and Magnus Wahlstr{\"o}m}
\institute{
  Royal Holloway, University of London\\
 Egham, Surrey TW20 0EX, UK\\
} 

\begin{document}

\maketitle

\begin{abstract}
In the Mixed Chinese Postman Problem (MCPP), given a weighted mixed graph $G$ ($G$ may have both edges and arcs), our aim is to find a minimum weight closed walk traversing each edge and arc at least once. The MCPP parameterized by the number of edges in $G$ or the number of arcs in $G$ is fixed-parameter tractable as proved by van Bevern {\em et al.} (in press) and Gutin, Jones and Sheng (ESA 2014), respectively.   In this paper, we consider the unweighted version of MCPP.
Solving an open question of van Bevern {\em et al.} (in press), we show that somewhat unexpectedly MCPP parameterized by the (undirected) treewidth of $G$ is W[1]-hard. In fact, we prove that even the MCPP parameterized by the pathwidth of $G$ is W[1]-hard.
On the positive side, we show that the unweighted version of MCPP parameterized by tree-depth is fixed-parameter tractable. 
We are unaware of any natural graph parameters between pathwidth and tree-depth and so our results provide a dichotomy of the complexity of MCPP. 
Furthermore, we believe that MCPP is the first problem known to be W[1]-hard with respect to treewidth but FPT with respect to tree-depth.
\end{abstract}

\section{Introduction}

A {\em mixed graph} is a graph that may contain both edges and arcs (i.e., directed edges). A mixed graph $G$ is {\em strongly connected} if for each ordered pair $x,y$ of vertices in $G$ there is a path from $x$ to $y$ that traverses each arc in its direction. In this paper, we will deal with simple mixed graphs\footnote{A simple mixed graph has at most one arc or edge between any pair of vertices.} and (possibly non-simple) directed multigraphs (with multiple arcs between each pair of vertices). However, whenever we refer to the treewidth (pathwidth, tree-depth) of a graph, we mean the treewidth (pathwidth, tree-depth) of the underlying undirected graph. 

In this paper, we study the following well-known problem. 

 \begin{quote}
\fbox{~\begin{minipage}{0.9\textwidth}
  {\sc Mixed Chinese Postman Problem (MCPP)} \nopagebreak
  
    \emph{Instance:} A strongly connected mixed graph $G = (V, E \cup A)$, with vertex set $V$, set $E$ of edges and set $A$ of arcs; a weight function $w: E \cup A \rightarrow \mathbb{N}_0$.
      
    \emph{Output:} A closed walk of $G$ that traverses each edge and arc at least once, of minimum weight.
\end{minipage}~}
  \end{quote}

There is numerous literature on various algorithms and heuristics for MCPP; for informative surveys, see \cite{vanBevern,BruckerLNCS1981,EiseltOR1995,PengCJOR1989}.
When $A = \emptyset$, we call the problem the {\sc Undirected Chinese Postman Problem (UCPP)}, and when $E = \emptyset$, we call the problem the {\sc Directed Chinese Postman Problem (DCPP)}.
It is well-known that {\sc UCPP} is polynomial-time solvable \cite{EdmondsMP1973} and so is {\sc DCPP} \cite{BeltramiNetworks1974,ChristofidesOmega1973,EdmondsMP1973}, 
but {\sc MCPP} is NP-complete, even when $G$ is planar with each vertex having total degree $3$ and all edges and arcs having weight $1$ \cite{PapadimitriouJACM1976}.
  It is therefore reasonable to believe that {\sc MCPP} may become easier the closer it gets to {\sc UCPP} or {\sc DCPP} and indeed when parameterized by the number of edges in $G$ or the number of arcs in $G$, {\sc MCPP} is proved to be fixed-parameter tractable (FPT, defined below) by van Bevern {\em et al.} \cite{vanBevern} and Gutin, Jones and Sheng \cite{GutinESA2014}, respectively.

In this paper, we consider structural parameterizations of {\sc MCPP}. 
van Bevern {\em et al.} \cite{vanBevern} noted that Fernandes, Lee and  Wakabayashi \cite{FernandesDAM2009} proved that {\sc MCPP} parameterized by the treewidth of $G$ is in XP (when all edges and arcs have weight $1$), and asked whether this parameterization of {\sc MCPP} is FPT.
 It is well-known that many graph problems are FPT when parameterized  by the treewidth of the input graph (only a few such problems are 
W[1]-hard; see, e.g., \cite{DomIWPEC2008,FellowsCOCAO2007,GolovachDAM2012}). In this paper, we show that somewhat unexpectedly the {\sc MCPP} parameterized by treewidth belongs to a small minority of problems, i.e., it is W[1]-hard. In fact, we prove a stronger result by (i) replacing treewidth with pathwidth, and (ii) assuming that all edges and arcs have weight $1$. 
 
To complement this, we show a positive result for the parameter tree-depth. 
We prove that {\sc MCPP} parameterized by tree-depth is FPT. {\sc MCPP} is unusual in this regard, as this is the first problem we are aware of which is FPT  parameterized by tree-depth but W[1]-hard parameterized by treewidth. 
Following \cite{FernandesDAM2009}, we assume that all weights equal 1, however, we do not foresee any significant difficulty in generalizing our result to the weighted case.



Our paper is organised as follows. In the rest of this section, we provide some basics definitions on parameterized complexity as well as the definitions of treewidth, pathwidth and tree-depth. 
In Section \ref{sec2}, we introduce 
\add{an intermediate problem {\sc Properly Balanced Subgraph (\PBS)},}
and give a W[1]-hardness proof for a restricted variant of it. In Section \ref{sec3}, we reduce this variant of \PBS{} into {\sc MCPP} parameterized by pathwidth, showing that the latter is also W[1]-hard.
In Section \ref{sec4} we show that \PBS{} is FPT with respect to tree-depth, as outlined  above, and in Section \ref{sec5} we reduce {\sc MCPP} parameterized by tree-depth to \PBS{} parameterized by tree-depth, showing that this parameterization of {\sc MCPP} is FPT.
We conclude the paper with Section \ref{sec6}, where, in particular, we mention an open question from  \cite{vanBevern} on another parameterization of {\sc MCPP}.

     For reasons of space, many proofs and figures are deferred to the Appendix.

\paragraph{Parameterized Complexity.}
A \emph{parameterized problem} is a subset $L\subseteq \Sigma^* \times
\mathbb{N}$ over a finite alphabet $\Sigma$. $L$ is
\emph{fixed-parameter tractable (FPT)} if the membership of an instance
$(I,k)$ in $L$ can be decided in time
$f(k)|I|^{O(1)}$ (called {\em FPT time} and the corresponding algorithm an {\em FPT algorithm}), 
and in XP if it can be decided in time $|I|^{f(k)}$, 
where $f$ is a computable function of the
{\em parameter} $k$ only. 
Many parameterized problems are believed not to be FPT; one of them is {\sc $\kappa$-Clique}, {\sc Clique} parameterized by the number $\kappa$ of vertices of the required clique. 
A problem with parameter $k$ is called $W[1]$-{\em hard} if {\sc $\kappa$-Clique} can be reduced to it in FPT time with respect to $\kappa$ such that $k$ is bounded by a computable function of $\kappa$. 
 For more information on parameterized algorithms and complexity, see \cite{Downey2013}.

 \paragraph{Treewidth, Pathwidth and Tree-depth.} 
 For an undirected graph $G = (V,E)$, a \emph{tree decomposition} of $G$ is a pair $({\cal T}, \beta)$, where ${\cal T}$ is a tree and $\beta:V({\cal T}) \rightarrow 2^V$ such that $\bigcup_{x \in V({\cal T})}\beta(x) = V$, for each edge $uv \in E$ there exists $x \in V({\cal T})$ with $u,v \in \beta(x)$, and for each $v \in V$ the set of nodes $\beta^{-1}(v)$ forms a connected subtree in ${\cal T}$. 
 The \emph{width} of $({\cal T}, \beta)$ is $\max_{x \in V({\cal T})}(|\beta(x)|-1)$.
 The \emph{treewidth} of $G$ is the minimum width of all tree decompositions of $G$.
  The \emph{pathwidth} of a graph is the minimum width of all tree decompositions $({\cal T}, \beta)$ for which ${\cal T}$ is a path.
  For a directed multigraph $H$, we will use $\pw(H)$ to denote the pathwidth of the underlying undirected graph of $H$.



The \emph{tree-depth} of a connected graph $G$ is defined as follows. 
Let $T$ be a rooted tree with vertex set $V(G)$, such that if $xy$ is an edge in $G$ then $x$ is either an ancestor or a descendant of $y$ in $T$. Then we say that \emph{$G$ is embedded in $T$}.
The \emph{depth} of $T$ is the number of vertices in a longest path in $T$ from the root to a leaf.
The \emph{tree-depth} of $G$ is the minimum $t$ such that $G$ is embedded in a tree of depth $t$.
Thus, for example, a star $K_{1,r}$ has tree-depth $2$. A path of length $n$ has tree-depth $O(\log n)$.
A graph of tree-depth $k$ has pathwidth at most $k$.



  \section{Properly Balanced Subgraph Problem}\label{sec2}

  In this section, we introduce the problem {\sc Properly Balanced Subgraph (\PBS)}, and show that it is W[1]-hard parameterized by pathwidth.
  In Section \ref{sec4}, we will show that a special case of the problem with restricted weights is fixed-parameter tractable with respect to tree-depth.
 
 A directed multigraph is called {\em balanced} if the in-degree of each vertex coincides with its out-degree. 
  A \emph{double arc} is a specified pair of arcs $(a,a')$ such that $a$ and $a'$ have the same heads and tails.
   We will say that a subgraph $D'$ of $D$ \emph{respects double arcs} if  $|A(D') \cap \{a,a'\}| \neq 1$ for every double arc $(a,a')$.
   That is, if $D'$ uses one arc of a double arc then it must use both of them.
   A \emph{forbidden pair} is a specified pair of arcs $(b,b')$ such that $b$ is the reverse of $b'$.
 We say that $D'$ \emph{respects forbidden pairs} if  $|A(D') \cap \{b,b'\}| \neq 2$ for every forbidden pair $(b,b')$.
 That is, if $D'$ uses one arc of a forbidden pair then it cannot use the other.
   We will say that a subgraph $D'$ of $D$ is \emph{properly balanced} if $D'$ is balanced 
 and respects double arcs and forbidden pairs.
PBS is then defined as follows.

   \begin{quote}
     \fbox{~\begin{minipage}{0.9\textwidth} {\sc Properly Balanced
           Subgraph (\PBS)} \nopagebreak

         \emph{Instance:} A directed multigraph $D = (V, A)$; a weight
         function $w: A \rightarrow \mathbb{Z}$; a set $X = \{(a_1,a_1'),
         \dots, (a_r, a_r')\}$ of \emph{double arcs} with $a_i,a_i'
         \in A$ for each $i \in [r]$;
         a set $Y = \{(b_1,b_1'), \dots, (b_s, b_s')\}$ of \emph{forbidden
           pairs} with $b_i,b_i' \in A$ for each $i \in [s]$.
         Double arcs are disjoint from forbidden pairs. 
      
         \emph{Output:}
         A properly balanced subgraph $D'$ of $D$ of negative weight, if
         one exists.
       \end{minipage}~}
   \end{quote}
  
 
 \subsection{Gadgets for \PBS}
 
 We now describe some simple gadget graphs
(for now we do not assign weights; we will do this later).
 Each gadget will have some number of \emph{input} and \emph{output} arcs. Later, we will combine these gadgets by joining the input and output arcs of different gadgets together using double arcs. Henceforth, for each positive integer $n$, $[n]=\{1,2,\ldots ,n\}$.

%
%
%
%
%
%

 A \emph{Duplication gadget} has one input arc and $t$ output arcs, for some positive integer $t$. The vertex set consists of vertices $x,y$, and $u_i, v_i$ for each $i \in [t]$. The arcs form a cycle $xyu_1v_1\dots u_tv_tx$. The input arc is the arc $xy$, and the output arcs are the arcs $u_iv_i$ for each $i \in [t]$. 

 A \emph{Choice gadget} has one input arc $xy$ and $t$ output arcs $u_iv_i : i \in [t]$, for some positive integer $t$.
 The vertex set consists of the vertices $x,y,z,w$ and $u_i,v_i$ for each $i \in [t]$. The arcs consist of a path $wxyz$, and the path $zu_iv_iw$ for each $i \in [t]$. 

 Finally, a \emph{Checksum gadget} has $t_l$ \emph{left input} arcs $x_iy_i: i \in [t_l]$ for some positive integer $t_l$, and $t_r$ \emph{right input} arcs $u_iv_i: i \in [t_r]$, and no output arcs. The vertex set consists of the vertices $w,z$ together with $x_i,y_i$ for each $i \in [t_l]$ and $u_i,v_i$ for each $i \in [t_r]$.
 The arc set consists of the path $wx_iy_iz$ for each $i \in [t_l]$, and $zu_iv_iw$ for each $i \in [t_r]$.
 \remove{(See Figure \ref{fig:checksumGadget}.)}
 
%
%

\add{
Proposition \ref{prop:gadgets} below is easy to prove and thus its proof is omitted.

\begin{proposition}\label{prop:gadgets}
 Let $D$ be one of the gadgets described above, and let $D'$ be a properly balanced subgraph of $D$.
  Then if $D'$ does not contain any input arcs of $D$, it does not contain any other arcs of $D$. Otherwise:  if $D$ is a Duplication gadget, then $D'$ contains every output arc of $D$; if $D$ is a Choice gadget, then $D'$ contains exactly one output arc of $D$;
  and if $D$ is a Checksum gadget, then $D'$ contains the same number of left input arcs as right input arcs of $D$.
Furthermore, for any choice of input and output arcs satisfying these conditions, there is subgraph $D'$ containing exactly those input and output arcs.
\end{proposition}
}

 Observe that in all of our gadgets, the vertices in input or output arcs all have in-degree and out-degree $1$.
 We next describe how to combine these gadgets. For two unjoined arcs $uv$ and $xy$ (possibly in disjoint graphs), the operation of \emph{joining} $uv$ and $xy$ is as follows: Identify $u$ and $x$, and identify $v$ and $y$. Keep both $uv$ and $xy$, and add $(uv,xy)$ as a double arc.
 

 \begin{lemma}\label{lem:joinCorrectness}
 Let $D_1$ and $D_2$ be disjoint directed multigraphs. Let $u_1v_1, \dots u_tv_t$ be arcs in $D_1$, and let $x_1y_1, \dots x_ty_t$ be arcs in $D_2$, such that $u_i$ and $v_i$ both have in-degree and out-degree $1$ in $D_1$, and $x_i$ and $y_i$ both have in-degree and out-degree $1$ in $D_2$, for each $i \in [t]$.
 Let $D$ be the graph formed by joining $u_iv_i$ and $x_iy_i$, for each $i \in [t]$. 
  Then a subgraph $D'$ of $D$ is a properly balanced graph if and only if
  \add{(1) $|A(D') \cap \{u_iv_i, x_iy_i\}| \neq 1$ for each $i \in [t]$; and (2) $D'$ restricted to $D_1$ is a properly balanced subgraph of $D_1$, and $D'$ restricted to $D_2$ is a properly balanced subgraph of $D_2$.}
%
%
 \end{lemma}

%
%

 \subsection{W[1]-hardness of \PBS}

 
 \add{By joining an output arc of one gadget to the input arc of another gadget, we have that a solution will only pass through the second gadget if it uses the corresponding arc of the first gadget. Thus for example, if a Duplication gadget has $k$ output arcs, each of which is joined to the input arc of a Choice gadget, then any solution that uses the input arc of the Duplication gadget has to use exactly one output arc from each of the Choice gadgets. By combining gadgets in this way, we can create ``circuits" that represent instances of other problems.

We will use this idea to represent the following W[1]-hard problem. 
In {\sc $k$-Multicolored Clique}, we are given a graph  $G = (V_1 \cup V_2 \dots \cup V_k, E)$, such that for each $i \in [k]$, $V_i$ forms an independent set, and asked to decide whether $G$ contains a clique with $k$ vertices, where $k$ is the parameter.
}
 
%
%
%
%
%
  
  \begin{theorem}\cite{FellowsTCS2009}
   {\sc $k$-Multicolored Clique} is W[1]-hard.
  \end{theorem}

  

%
 
 \begin{theorem}\label{thm:mbsdaHardness}
  \PBS{} is W[1]-hard parameterized by pathwidth, 
  even when there are no forbidden arcs, there is a single arc $a^*$ of weight $-1$ and $a^*$ is not part of a double arc, and all other arcs have weight $0$.
 \end{theorem}
 
 \add{We give a sketch of the proof (for a full proof, see the Appendix):
 
Let $G = (V_1 \cup \dots \cup V_k, E)$ be an instance of {\sc $k$-Multicolored Clique}.
We construct an equivalent instance of \PBS{} as follows.

Initially we have a duplication gadget with $k$ output arcs, whose input arc is the only arc of weight $-1$. All other arcs will have weight $0$. Thus, any solution to the PBS instance will have to use this duplication gadget and all its output arcs. 
Then for each $i \in [k]$, we choose a vertex $v_i \in V_i$ (represented by a Choice gadget with $|V_i|$ output arcs, whose input arc is joined to the initial Duplication gadget).
For each choice of $v_i$, and for each $j \in [k] \setminus \{i\}$ (enforced by a Duplication gadget with $k-1$ output arcs), we then choose an edge $e_{i \rightarrow j}$ that is adjacent to $v_i$ and a vertex in $V_j$ (represented by a choice gadget with $|N(v_i)\cap V_j|$ output arcs).

The graph so far looks like a ``tree'' of gadgets, and as such has bounded treewidth. Is easy to show that it also has bounded pathwidth. It enforces that we choose a set of vertices $v_1, \dots , v_k$, and then an edge $e_{i \rightarrow j}$ for each ordered pair $(i,j), i \neq j$.
(Each possible choice for $e_{i \rightarrow j}$ is represented by one output arc on the last layer of Choice gadgets). 
Now observe that $v_1, \dots , v_k$ forms a clique if and only if there are choices of $e_{i \rightarrow j}$ such that $e_{i \rightarrow j} = e_{j \rightarrow i}$ for each $(i,j)$.

We can check for this condition as follows. Firstly, we associate each edge $e$ with a unique number $n_e$. Then, for each output arc corresponding to the edge $e$, we join that arc to a Duplication gadget with $n_e$ output arcs. (This increases the pathwidth of the graph by a constant). Then for each unordered pair 
$\{i,j\}, i < j$, we create a Checksum gadget $\textsc{CheckEdge}(i,j)$. The left input arcs of this gadget are joined to all the output arcs of all Duplication gadgets corresponding to a choice for $e_{i \rightarrow j}$, and the right input arcs are joined to all the output arcs of all Duplication gadgets corresponding to a choice for $e_{j \rightarrow i}$. This completes the construction of the graph.

It follows that for any solution to the PBS instance, the number of left input arcs of $\textsc{CheckEdge}(i,j)$ in the solution is equal to the number associated with the edge chosen for  $e_{i \rightarrow j}$ . 
Similarly the number of right input arcs in the solution is equal to the number associated with the edge chosen for $e_{j \rightarrow}$.
As these numbers have to be equal, it follows that there is a solution if and only if the choice for $e_{i \rightarrow j}$ is the same as the choice for $e_{j \rightarrow i}$ for each $i \neq j$.
Thus, our PBS instance has a solution of negative weight if and only if $G$ has a clique. It remains to check the pathwdith of the graph.

Before the addition of the Checksum gadgets, the graph has pathwdith bounded by a constant. As the input arcs of these gadgets are joined to other arcs, adding the Checsum gadgets only requires adding $2$ vertices for each $(i,j)$. Thus, the pathwdith of the constructed graph is $O(k^2)$.

 }

  \section{Reducing \PBS{} to \MCPP}\label{sec3}

  We now show how to reduce an instance of \PBS{}, of the structure given in Theorem~\ref{thm:mbsdaHardness}, to \MCPP.
  Let ($D = (V,A), w, X=\{(a_i,a_i'): i \in [t]\}, Y=\emptyset$) be an instance of \PBS{} with double arcs $X$ and no forbidden pairs, and where $w(a^*)=-1$ for a single arc $a^*$ and $w(a)=0$ for every other arc. We may assume that $a^*$ is not in a double arc. 
  We will produce an instance $G$ of \MCPP{} and an integer $W$, such that $G$ has a solution of weight $W$, and $G$ has a solution of weight less than $W$ if and only if our instance of \PBS{} has a solution with negative weight.
  All edges and arcs in our \MCPP{} instance will have weight $1$.
  
  \add{
  We derive $G$ by replacing every subgraph of $D$ between a pair of vertices (either a double arc, a single arc of weight $0$, or a single arc of weight $-1$) by an appropriate gadget.
The gadgets will be such that within each gadget, there are only two \MCPP{} solutions of reasonable weight: a solution that is balanced within the gadget (corresponding to not using the original arc/double arc in a solution to $D$); and a solution that is imbalanced at the vertices by the same amount that the original arc / double arc is (which corresponds to using the original arc / double arc in a solution to $D$).
Thus, every properly balanced subgraph of $D$ corresponds to a balanced subgraph of $G$, and vice versa.

For each gadget, except the gadget corresponding to the negative weight arc, the weights of the two solutions will be the same.
In the case of the negative weight arc, the solution that correponds to using te arc will be cheaper by $1$. 
Thus, there are two possible weights for a solution to $G$, and the cheaper weight is only possible if $D$ has a properly balanced subgraph of negative weight.

  }
  
  In what follows, we will construct arcs and edges of two weights, \emph{standard} and \emph{heavy}. Standard arcs and edges have weight 1; heavy arcs and edges have weight $M$, where $M$ is 
  \remove{a large number to be determined later, chosen so that no solution of weight at most $W$ can traverse a heavy arc or edge more than once.}
  \add{a large enough number that we may assume that no solution traverses a heavy arc or edge more than once.}
  This will be useful to impose structure on the possible solutions when constructing gadgets. 
  \remove{At the end of the construction, we will show how to replace the heavy arcs and edges by paths of standard arcs and edges, so that the end result is an unweighted graph. }
  \add{A heavy arc (edge) is equivalent to a directed (undirected) path of length $M$, and so we also show W[1]-hardness for the unweighted case.}

 \add{The gadgets are constructed as follows. It is straightforward to verify that each gadget has only two solutions that traverse each arc or edge exactly once, and that the imbalances and weights of these solutions are as described above.
(For a full proof, see the Appendix.)}
 
 {\bf For an arc $uv$ of weight $0$ that is not part of a double arc:} Construct {\sc Gadget($u,v$)} by creating
%
%
%
a new vertex $z_{uv}$, with standard arcs $z_{uv}u$ and $z_{uv}v$, two heavy arcs $uz_{uv}$, and a heavy arc $vz_{uv}$. 

 {\bf For an arc $uv$ of weight $-1$ that is not part of a double arc:}
%
  Construct {\sc Gadget($u,v$)} by adding
  two new vertices $w_{uv}$ and $z_{uv}$, with standard arcs $z_{uv}u, z_{uv}w_{uv}$ and $vw_{uv}$, two heavy arcs $uz_{uv}$, one heavy arc $w_{uv}z_{uv}$, and two heavy arcs $w_{uv}v$. 
 
 {\bf For a double arc from $u$ to $v$:}
 {\sc Gadget($u,v$)} consists of a heavy arc $uv$ and a heavy edge $\{u,v\}$. 
 Assuming a solution traverses each heavy arc/edge exactly once, the only thing to decide is in which direction to traverse the undirected edge.

\remove{

{\bf Removing heavy arcs and edges.} Finally, we replace every heavy arc (edge, respectively) we just created, of weight $M$, by a directed (undirected) path of length $M$, where all internal vertices have degree two. 
Note that in any minimal solution, each arc or edge in such a path will be traversed the same number of times
(if one edge in the path is traversed more times than its neighbor, then it must be traversed at least twice more, including at least once in each direction, and so the solution is not minimal).
Thus in the analysis, we may treat such a path as effectively being a single edge or arc of weight $M$.
}

\add{We note that each of our gadgets has pathwidth bounded by a constant. It can be shown that replacing the arcs of $D$ with gadgets in this way will only increase the pathwdith by a constant.}

We now have that, given an instance $(D,w,X, Y)$  of \PBS{} of the type specified in Theorem \ref{thm:mbsdaHardness}, we can in polynomial time create an equivalent instance $G$ of \MCPP{} with pathwidth bounded by 
\add{$O(\pw(D))$.} 
We therefore have a parameterized reduction from this restriction of \PBS, parameterized by pathwidth, to \MCPP{} parameterized by pathwidth. 
As this restriction of \PBS{} is W[1]-hard by Theorem \ref{thm:mbsdaHardness}, we have the following theorem.

\begin{theorem}\label{thm:mcppTreewidthHardness}
 \MCPP{} is W[1]-hard parameterized by pathwidth.
\end{theorem}

\section{\PBS{} Parameterized by Tree-depth}\label{sec4}

In this section we show that a certain restriction of \PBS{} is fixed-parameter tractable with respect to tree-depth. 
The restriction we require is that all arcs in double arcs have weight $0$, all arcs in forbidden pairs have weight $-1$, and all other arcs have weight $1$ or $-1$.
We choose this restriction, as this is the version of \PBS{} that we get when we reduce from \MCPP.

Aside from the familiar tool of Courcelle's theorem (cf. \cite{FlumGroheBook}), our main technical tool is Lemma \ref{lem:subgraphBound}, which shows that we may assume there exists a solution with size bounded by a function of tree-depth. The following simple observation will be useful in the proof of Lemma~\ref{lem:subgraphBound}.

\begin{lemma} \label{lem: combiningImbalances}
 Let $\{H_i: i \in {\cal I}\}$ be a family of pairwise arc-disjoint subgraphs of $G$, such that each $H_i$ respects double arcs. Then $H = \bigcup_{i \in {\cal I}} H_i$ is a properly balanced subgraph of $G$ if and only if $H$ is balanced and $H$ respects forbidden pairs.
\end{lemma}

We are now ready to prove that any properly balanced subgraph decomposes into properly balanced subgraphs of size bounded by a function of tree-depth. This will allow us to assume, in Theorem \ref{thm:balancedSubgraphFPT}, that a solution has bounded size.

\begin{lemma}\label{lem:subgraphBound}
 Let $G$ be a directed multigraph (with double arcs and forbidden pairs) of tree-depth $k$, and let $H$ be a properly balanced subgraph of $G$. Then $H$ is a union of pairwise arc-disjoint graphs $H_i$, each of which is 
a properly balanced subgraph of $G$, with $|A(H_i)| \le f(k)$ where $f(k) = 2^{2^k}$.
\end{lemma}
\begin{proof}
We prove the claim by induction on the tree-depth $k$. For the base case, observe that if $k=1$ then $G$ has no arcs, and the claim is trivially true. 
So now assume that $k \ge 2$, and that the claim holds for all graphs of tree-depth less than $k$.
We also assume that $H$ is 2-connected, as otherwise a block decomposition of $H$ is a decomposition into properly balanced subgraphs, and we may apply our result to each block of $H$.
Similarly, if $G$ is not 2-connected but $H$ is, then $H$ lies inside one block of $G$, and we may restrict our attention to this block. 
Hence assume that $G$ is 2-connected as well.

 Let $G$ be embedded in a tree $T$ of depth $k$, and let $x$ be the root of $T$. Observe that $x$ has only one child in $T$, as otherwise $x$ is a cut-vertex in $G$.
 Let $y$ be this child, and let $G'$ be the multigraph derived from $G$ by identifying $x$ and $y$.
 Similarly, let $H'$ be the subgraph of $G'$ derived from $H$ by identifying $x$ and $y$. 
 Observe that $H'$ is balanced as $H$ is balanced and so the number of arcs into $\{x,y\}$ equals the number of arcs out of it. 
 Let $B$ be the set of arcs in $H$ between $x$ and $y$, and
 observe that there is a one-to-one correspondence between the arcs of $H'$ and the arcs of $H$ not in $B$.
 By identifying $x$ and $y$ in $T$, we get that $G'$ has tree-depth at most $k-1$.
 
 By the induction hypothesis, $H'$ can be partitioned into a family $\{H_i': i \in {\cal I}'\}$ of pairwise arc-disjoint properly balanced subgraphs of $G'$, each having at most $f(k-1)$ arcs.
 For each $i \in {\cal I}'$, let $F_i$ be the subgraph of $G$ corresponding to $H_i'$.
 Observe that $B$ can also be partitioned into a family $\{F_i: i \in {\cal I}''\}$ of subgraphs with at most $2$ arcs, that respect double arcs (we add any double arc from $B$ as a subgraph $F_i$, and add every other arc as a single-arc subgraph).
 
 Letting ${\cal J} = {\cal I}' \cup {\cal I}''$, we have that $\{F_i: i \in {\cal J}\}$ is a partition of $H$, each $F_i$ has at most $f(k-1)$ arcs, and each $F_i$ respects double arcs and is balanced everywhere except possibly at $x$ and $y$. We now combine sets of these subgraphs into subgraphs that are balanced everywhere.
 
 For each $i \in {\cal J}$, let $t_i$ be the imbalance of $F_i$ at $x$, i.e. $t_i = d_{F_i}^+(x) - d_{F_i}^-(x)$. 
 Observe that $|t_i|\le \frac{f(k-1)}{2}$ for each $i$ and, as $H$ is balanced, $\sum_{i \in {\cal J}}t_i = 0$.
We now show that there exists a subset ${\cal J}' \subseteq {\cal J}$ such that $|{\cal J}'| \le f(k-1)-1$ and $\sum_{i \in {\cal J}}t_i = 0.$
 To see this, 
 let ${\cal J}_1$ be a set containing a single $t_i$, of minimum absolute value, 
 and iteratively construct sets ${\cal J}_r$ by adding $i$ such that $t_i < 0$ to  ${\cal J}_{r-1}$ if $\sum_{p \in {\cal J}_{r-1}}t_p > 0$, and adding $i$ such that $t_i > 0$ otherwise. 
 Now note that either $t_i=\pm f(k-1)/2$ for every $i \in {\cal J}$, in which case we have a subset ${\cal J}'$ with $|{\cal J}'|=2$, 
or $|\sum_{i \in {\cal J}_r}t_i| < \frac{f(k-1)}{2}$ for each $r$, 
and therefore there are at most $f(k-1)-1$ possible values that $\sum_{i \in {\cal J}_r}t_i$ can take. 
 Then there exist $r,r'$ such that $r'<r$, $r-r' \le f(k-1)-1$, and  $\sum_{i \in {\cal J}_r\setminus {\cal J}_{r'}}t_i = 0$. 
So let ${\cal J}'={\cal J}_r\setminus {\cal J}_{r'}$. 
 
 Now let $H_1 = \bigcup_{i \in {\cal J}'}F_i$. Then by construction, $H_1$ is balanced at every vertex (as it is balanced for every vertex other than $y$, and a directed multigraph graph cannot be imbalanced at a single vertex), and $H_1$ respects double arcs. As $H_1$ is a subgraph of $H$, $H_1$ also respects forbidden pairs. Therefore $H_1$ is a properly balanced subgraph, with number of arcs at most $(f(k-1)-1)f(k-1)$.
Observe that $f(k)=2^{2^k}$ is a solution to the recursion $(f(k-1)-1)f(k-1)<f(k)$ with $f(1)=4$. Thus, $H_1$ has at most $2^{2^k}$ arcs, as required.

 By applying a similar argument to $H \setminus H_1$, we get a properly balanced subgraph $H_2$ with at most $f(k)$ arcs. 
 Repeating this process, we get a partition of $H$ into properly balanced subgraphs each with at most $f(k)$ arcs.
 \qed
\end{proof}

Using Lemma \ref{lem:subgraphBound}, we may now assume that if $G$ has a properly balanced subgraph with negative weight, then it has a properly balanced subgraph of negative weight with at most $f(k)$ arcs (as any negative weight properly balanced subgraph can be partitioned into properly balanced subgraphs of at most $f(k)$ arcs, at least one of which must have negative weight).

\subsection{Fixed-Parameter Tractability of \PBS}

As the tree-depth of $G$ is at most $k$, it follows that it has pathwidth at most $k-1$ \cite{BodlaenderALGO1995}.
Using this fact, and the fact that we may assume that a solution has at most $f(k)$ arcs, we have the following:

\begin{theorem}\label{thm:balancedSubgraphFPT}
 \PBS{} is FPT with respect to tree-depth, provided that 
all arcs in double arcs have weight $0$, all arcs in forbidden pairs have weight $-1$, and all other arcs have weight $1$ or $-1$.
\end{theorem}

\begin{proof}


Let $D$ be an instance of \PBS{} and let $k$ be the tree-depth of $D$.
By Lemma \ref{lem:subgraphBound}, we may assume that a solution has at most $f(k)$ arcs.
We will prove the theorem by guessing all possible solutions $D'$ (up to isomorphism), and then checking whether $D'$ is isomorphic to a subgraph of $D$ using Courcelle's theorem.


First, we assign each arc $a$ of $D$ the label 
{\sc Double} if $a$ is in a double arc, the label {\sc Forbidden} if $a$ is in a forbidden pair,
and otherwise the label
{\sc Negative} if $a$ has weight $-1$ and  {\sc Positive} if $a$ has weight $1$.
Next, enumerate all balanced directed multigraphs $D'$ with at most $f(k)$ arcs, together with all possible labelings $\lambda:(A(D')) \rightarrow \{\textsc{Negative}, \textsc{Positive}, \textsc{Double}, \textsc{Forbidden}\}$. 
This can be done in time $O^*((5f(k)^2)^{f(k)})$.

For each constructed directed multigraph $D'$, we first check whether $D'$ would be a properly balanced subgraph of $D$, if it existed as a subgraph of $D$. This can be done by checking, for each pair of vertices, that $D'$ has at most one arc labelled {\sc Forbidden}, and either $0$ or $2$ arcs (in the same direction) labelled {\sc Double}.
If $D'$ is not a properly balanced subgraph, we disregard it. Otherwise, we then check whether $D'$ has negative weight. (This can be done by checking that the number of {\sc Negative} and {\sc Forbidden} arcs is greater than the number of {\sc Positive} arcs.)

If $D'$ has negative weight, it remains to check whether $D'$ is isomorphic to a subgraph of $D$ (respecting the labels).
As $D$ has treewidth at most $k$ and $D'$ has at most $f(k)$ vertices, this can be done in fixed-parameter time by Courcelle's theorem  
(see Theorem 11.37 in \cite{FlumGroheBook}).

\qed
\end{proof}

\section{Positive Result: Reducing \MCPP{} to \PBS} \label{sec5}

In this section, we consider \MCPP{} with all weights equal 1 parameterized by tree-depth.
In contrast to pathwidth, we will show that \MCPP{} parameterized by tree-depth is FPT. Hereinafter, $b_H(v)$ denotes the imbalance of $H$, i.e. $d_H^+(v)-d_H^-(v)$.
\add{In the problem \compMCPP, we are given an instance of \MCPP{} together with a solution $H$, and asked to find a solution $H'$ of weight less than $w(H)$, if one exists. }
To solve an instance of \MCPP, it would be enough to find some (not necessarily optimal) solution of weight $M$, then repeatedly apply \remove{the following problem} \compMCPP{} to find better solutions, until we find a solution which cannot be improved by \compMCPP{} and is therefore optimal. As \compMCPP{} returns an improved solution if one is available, we would have to apply \compMCPP{} at most $M$ times.

%
%

 To show that our approach leads to an FPT algorithm for \MCPP, we first show that we may assume that $M$ is bounded by an appropriate value.
 
 \begin{lemma}\label{lem:naiveCPPBound}
  Given an instance $(G,w)$ of \MCPP{} with $m$ arcs and edges, we can, in polynomial time, find closed walk of of $G$ that traverses each edge and arc at least once, if such a walk exists, and this walk traverses each arc at most $m+1$ times. 
 \end{lemma}

As with the hardness proof, we will use \PBS{} as an intermediate problem.
We now reduce \compMCPP{} to \PBS, in the following sense: For any input graph $G$ and initial solution $H$, we produce a directed multigraph $D$ (with double arcs and forbidden pairs), such that $D$ has a properly balanced subgraph of negative weight if and only if $G$ has a solution of weight less than $w(H)$.

For any adjacent vertices $u,v$ in $G$, let $G_{uv}$ be the subgraph of $G$ induced by $\{u,v\}$, and similarly let $H_{uv}$ be the subgraph of $H$ induced by $\{u,v\}$. Let $M = w(H)$. Thus, we may assume that any improved solution has weight less than $M$.
By Lemma \ref{lem:naiveCPPBound} and the assumption that the weight of every arc and edge is 1, we may assume $M \le m^2+m$.
 

For each edge and arc $uv$ in $G$, we will produce a gadget $D_{uv}$, based on $G_{uv}$ and $H_{uv}$ and the value $M$. The gadget $D_{uv}$ is a directed multigraph, possibly containing double arcs or forbidden pairs, and by combining all the gadgets, we will get an instance $D$ of \PBS. 


We now construct $D_{uv}$ according to the following cases (roughly speaking, a positive weight arc represents adding an arc in that direction, and a negative weight arc represents removing an arc in the opposite direction):

{\bf If $G_{uv}$ is an arc from $u$ to $v$ and $H_{uv}$ traverses $uv$  $t \le M$ times:} Then $D_{uv}$ has $t-1$ arcs from $v$ to $u$ of weight $-1$, and $M-t$ arcs from $u$ to $v$ of weight $1$.

 {\bf If $G_{uv}$ is an edge between $u$ and $v$, and $H_{uv}$ traverses $uv$ from $u$ to $v$ $t \le M$ times, and from $v$ to $u$ $0$ times:}
Then $D_{uv}$ has a double arc $(a,a')$, where $a$ and $a'$ are both arcs from $v$ to $u$ of weight $0$. In addition, $D_{uv}$ has $t-1$ arcs from $v$ to $u$ of weight $-1$, $M-t$ arcs from $u$ to $v$ of weight $1$, and $M-1$ arcs from $v$ to $u$ of weight $1$.

 {\bf If $G_{uv}$ is an edge between $u$ and $v$ and $H_{uv}$ traverses $uv$ from $u$ to $v$ $t>0$ times, and from $v$ to $u$ $s>0$ times:}
 Then we may assume $s=t=1$, as otherwise we may remove a pair of arcs $(uv,vu)$ from $H$ and get a better solution to \MCPP{}. 
 Then $D_{uv}$ has $M-1$ arcs from $u$ to $v$ of weight $1$, $M-1$ arcs from $v$ to $u$ of weight $1$, and a forbidden pair $(a,a')$, where $a$ is an arc from $u$ to $v$, $a'$ is an arc from $v$ to $u$, and both $a$ and $a'$ have weight $-1$.



\begin{lemma}\label{lem:changeGadgets}
  Let $uv$ be an edge or arc in $G$, and let $B$ and $W$ be arbitrary integers such that $w(H_{uv})+W \leq M$. 
  Then the following are equivalent.
\begin{enumerate}
 \item\label{condH} There exists a graph $H'_{uv}$ with vertex set $\{u,v\}$ that covers $G_{uv}$, 
 such that  $w(H'_{uv}) = w(H_{uv})+W$ and $b_{H'_{uv}}(u) = b_{H_{uv}}(u) + B$;
 \item\label{condD} $D_{uv}$ has a subgraph $D'_{uv}$ which respects double arcs and forbidden pairs, such that $w(D'_{uv}) = W$ and $b_{D'_{uv}}(u) = B$.
\end{enumerate}
\end{lemma}

Note that in a graph $H''$ with two vertices $u$ and $v$, $b_{H''}(u) = - b_{H''}(v)$. Thus, in addition to implying that $b_{D'_{uv}}(u) =  b_{H'_{uv}}(u) - b_{H_{uv}}(u)$, the claim also implies that 
$b_{D'_{uv}}(v) = b_{H'_{uv}}(v) - b_{H_{uv}}(v)$.



\begin{lemma}\label{lem:MCPPtoPBSCorrectness}
 Let $D$ be the directed multigraph derived from $G$ and $H$ by taking the vertex set $V(G)$ and adding the gadget $D_{uv}$ for every arc and edge $uv$ in $G$.
 Then there exists a solution $H'$ with weight less than $H$ if and only if $D$ has a properly balanced subgraph of weight less than $0$.
\end{lemma}

\add{Lemma \ref {lem:MCPPtoPBSCorrectness} implies that we have a parameterized reduction from \compMCPP{} parameterized by tree-depth to \PBS{} parameterized by tree-depth.}
Then by Theorem \ref{thm:balancedSubgraphFPT}, we have the following theorem.

\begin{theorem}\label{thm:mcppFPT}
 \MCPP{} with all weights equal 1 is FPT with respect to tree-depth.
\end{theorem}

\section{Discussion}\label{sec6}

In this paper, we proved that {\sc MCPP} parameterized by pathwidth is W[1]-hard even if all edges and arcs of input graph $G$ have weight $1$.  This solves the second open question  of van Bevern {\em et al.} \cite{vanBevern} on parameterizations of {\sc MCPP}; the first being the parameterization by the number of arcs in $G$, which was proved to be FPT in \cite{GutinESA2014}.

We also showed that the unweighted version of {\sc MCPP} is FPT with respect to tree-depth.
This is the first problem we are aware of that has been shown to be W[1]-hard with respect to treewidth but FPT with respect to tree-depth.
The pathwidth of a graph lies between its treewidth and tree-depth.
The problem we studied is W[1]-hard with respect to both treewidth and pathwidth.
It is an open question whether there is a natural parameterized problem which is W[1]-hard with respect to treewidth but FPT with respect to pathwidth.

We call a vertex $v$ of $G$ {\em even} if the total number of arcs and edges incident to $v$ is even.
Another parameterization of {\sc MCPP} in \cite{vanBevern} is motivated by the fact that if each vertex of $G$ is even, then  {\sc MCPP} is polynomial-time solvable \cite{EdmondsMP1973}.  van Bevern {\em et al.} \cite{vanBevern} ask whether {\sc MCPP} parameterized by the number of non-even vertices is FPT.

\vspace{3mm}

\noindent{\bf Acknowledgement.} Research of GG was supported by Royal Society Wolfson Research Merit Award.

 \newpage
 
 \appendix
\section{Omitted proofs and figures from Section \ref{sec2}} \label{appendixA}



We begin by giving figures for the Choice and Checksum gadgets described in Section \ref{sec2} (a Duplication gadget is simply an even cycle with each even numbered arc being the input arc or an output arc).

 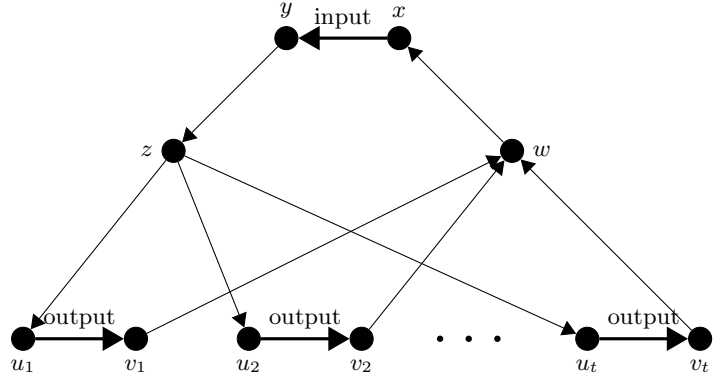
\begin{figure}[h]
  \centering
    \begin{tikzpicture}[>= triangle 60]

\node at (-1,4)(X)[city] [label=above: $x$]{} ;
\node at (-2.5,4)(Y)[city] [label=above: $y$] {}
edge[<-, very thick] node[auto, above]{input} (X)
;
\node at (-4,2.5)(Z)[city][label=left:$z$]{}
edge[<-] (Y);
\node at (0.5,2.5)(W)[city][label=right:$w$]{}
edge [->] (X);

\node at (-6,0)(U1)[city][label=below:$u_1$]{}
edge[<-](Z);
\node at (-4.5,0)(V1)[city][label=below:$v_1$]{}
edge[<-, very thick] node[auto, above] {output}(U1)
edge[->](W);
\node at (-3,0)(U2)[city][label=below:$u_2$]{}
edge[<-](Z);
\node at (-1.5,0)(V2)[city][label=below:$v_2$]{}
edge[->](W)
edge[<-, very thick] node[auto, above]{output}(U2);

\node at (0,0)(dots){\Huge $\dots$};

\node at (1.5,0)(Ut)[city][label=below:$u_t$]{}
edge[<-](Z);
\node at (3,0)(Vt)[city][label=below:$v_t$]{}
edge[->](W)
edge[<-, very thick] node[auto, above]{output}(Ut);

\end{tikzpicture}
\caption{A Choice gadget with $t$ output arcs. A balanced subgraph that contains the input arc will contain exactly one output arc.}\label{fig:choiceGadget}
\end{figure}
 
 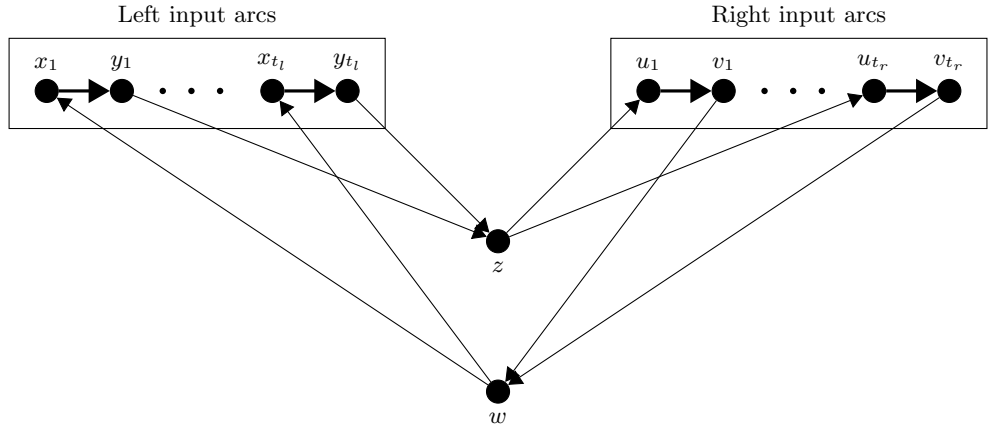
\begin{figure}[h]
  \centering
\begin{tikzpicture}[>= triangle 60]
\node at(0,0)(W)[city] [label=below: $w$]{};
\node at(0,2)(Z)[city] [label=below: $z$]{};

\node at (-6,4) (X1)[city] [label = above: $x_1$]{}
edge [<-] (W);
\node at (-5,4) (Y1)[city] [label = above: $y_1$]{}
edge [<-, very thick] (X1)
edge [->] (Z);

\node at (-4,4) (dots) {\Huge $\dots$};

\node at (-3,4) (Xtl)[city] [label = above: $x_{t_l}$]{}
edge [<-] (W);
\node at (-2,4) (Ytl)[city] [label = above: $y_{t_l}$]{}
edge [<-, very thick] (Xtl)
edge [->] (Z);

\draw (-6.5,3.5) rectangle (-1.5,4.7) [label=above: Left input arcs];
\node at (-4,5){Left input arcs};

\node at (2,4) (U1)[city] [label = above: $u_1$]{}
edge [<-] (Z);
\node at (3,4) (V1)[city] [label = above: $v_1$]{}
edge [<-, very thick] (U1)
edge [->] (W);

\node at (4,4) (dots) {\Huge $\dots$};

\node at (5,4) (Utr)[city] [label = above: $u_{t_r}$]{}
edge [<-] (Z);
\node at (6,4) (Vtr)[city] [label = above: $v_{t_r}$]{}
edge [<-, very thick] (Utr)
edge [->] (W);

\draw (1.5,3.5) rectangle (6.5,4.7) [label=above: Left input arcs];
\node at (4,5){Right input arcs};

\end{tikzpicture}
  \caption{A Checksum gadget with $t_l$ left input arcs and $t_r$ right input arcs.  A balanced subgraph will contain the same number of left and right input arcs.}\label{fig:checksumGadget}
\end{figure}

Recall Proposition \ref{prop:gadgets}:

\medskip
{\bf \noindent Proposition \ref{prop:gadgets}.}
{\em  Let $D$ be one of the gadgets described above, and let $D'$ be a properly balanced subgraph of $D$.
  Then if $D'$ does not contain any input arcs of $D$, it does not contain any other arcs of $D$. Otherwise:  if $D$ is a Duplication gadget, then $D'$ contains every output arc of $D$; if $D$ is a Choice gadget, then $D'$ contains exactly one output arc of $D$;
  and if $D$ is a Checksum gadget, then $D'$ contains the same number of left input arcs as right input arcs of $D$.
Furthermore, for any choice of input and output arcs satisfying these conditions, there is subgraph $D'$ containing exactly those input and output arcs.}
\medskip

We now give a proof of Lemma \ref{lem:joinCorrectness}.

\medskip
{\bf \noindent Lemma \ref{lem:joinCorrectness}.}
{\em
 Let $D_1$ and $D_2$ be disjoint directed multigraphs. Let $u_1v_1, \dots u_tv_t$ be arcs in $D_1$, and let $x_1y_1, \dots x_ty_t$ be arcs in $D_2$, such that $u_i$ and $v_i$ both have in-degree and out-degree $1$ in $D_1$, and $x_i$ and $y_i$ both have in-degree and out-degree $1$ in $D_2$, for each $i \in [t]$.
 Let $D$ be the graph formed by joining $u_iv_i$ and $x_iy_i$, for each $i \in [t]$. 
  Then a subgraph $D'$ of $D$ is a properly balanced graph if and only if:
  \add{(1) $|A(D') \cap \{u_iv_i, x_iy_i\}| \neq 1$ for each $i \in [t]$; and (2) $D'$ restricted to $D_1$ is a properly balanced subgraph of $D_1$, and $D'$ restricted to $D_2$ is a properly balanced subgraph of $D_2$.}
  }
\medskip
  
 \begin{proof} 
   Suppose first that $D'$ is a properly balanced subgraph of $D$. The first condition holds by definition, since $(u_iv_i, x_iy_i)$ is a double arc in $D$ for each $i \in [t]$. For the second condition, let $D_1'$  be $D'$ restricted to $D_1$. It is clear that $D_1'$ respects double and forbidden arcs in $D_1$. Furthermore, every vertex except $u_i$ or $v_i$, $i \in [t]$, has the same in- and out-degree in $D_1'$ as in $D'$, hence $D_1'$ is balanced for these vertices. Finally, for a vertex $u_i$ or $v_i$, $i \in [t]$, note that such a vertex has in- and out-degree 2 in $D$, and by the double arc $(u_iv_i, x_iy_i)$ either both these arcs are used in $D'$, or neither. In both cases, we find that the restriction $D_1'$ is balanced. Hence $D_1'$ is balanced at every vertex, and respects both double and forbidden arcs, i.e., $D_1'$ is a properly balanced subgraph of $D_1$. An analogous argument holds for $D_2$.

  Conversely, suppose that $|A(D') \cap \{u_iv_i, x_iy_i\}| \neq 1$ for each $i \in [t]$, $D'$ restricted to $D_1$ is a properly balanced subgraph of $D_1$, and $D'$ restricted to $D_2$ is a properly balanced subgraph of $D_2$. 
  Then by construction 
  $D'$ respects double arcs and forbidden arcs.
  As $D_1, D_2$ partition the arcs of $D$, and $D'$ is balanced when restricted to either of these graphs, we have that $D'$ is balanced. Thus, $D'$ is a properly balanced subgraph of $D$, as required.
  \qed
 \end{proof}

 The following technical lemma will also be useful in the proof of Theorem \ref{thm:mbsdaHardness}.

 \begin{lemma}\label{lem:joinPathwidth}
 Let $D', D_1, D_2, \dots, D_l$ be disjoint directed multigraphs, let $a_1, \dots, a_l$ be distinct arcs in $D'$, and let $D$ be a graph formed by joining the arc $a_i$ to an arc in $D_i$, for each $i \in [l]$.
 Then $\pw(D) \le \pw(D') + \max_i(\pw(D_i)) + 1$.
 \end{lemma}
 \begin{proof}
  Consider a minimum width path decomposition of $D'$. For each $i \in [l]$, let $x_i$ be the bag in the path decomposition of $D'$ that contains both vertices of $a_i$ (if there is a choice of bags, let $x_i$ be the bag of smallest size). Now replace $x_i$ with two identical bags $x_i',x_i''$, and in between $x_i'$ and $x_i''$ add a sequence of bags formed by taking a minimum width path decomposition of $D_i$ and adding all the vertices of $x_i$ to each bag.
  Do this for each $i \in [l]$. The resulting decomposition is a path decomposition of $D$. By construction and by choice of $x_i$, the width of this decomposition is  at most $\pw(D') + \max_i(\pw(D_i)) + 1$.
  \qed
 \end{proof}
 
 We are now ready to give a full proof of Theorem \ref{thm:mbsdaHardness}.

  
  \medskip
{\bf \noindent Theorem \ref{thm:mbsdaHardness}.}
{\em \PBS{} is W[1]-hard parameterized by pathwidth, 
  even when there are no forbidden arcs, there is a single arc $a^*$ of weight $-1$ and $a^*$ is not part of a double arc, and all other arcs have weight $0$.
}
\medskip
  
  \begin{proof}
  Given an instance $G = (V_1 \cup V_2 \dots \cup V_k, E)$ of {\sc $k$-Multicolored Clique}, 
  let $e_1, \dots e_m$ be an arbitrary enumeration of the edges of $E$. For each unordered pair $\{i,j\} \subseteq [k]$ with $i\neq j$, let $E_{\{i,j\}}$ be the subset of edges in $E$ with one vertex in $V_i$ and the other in $V_j$. Note that any $k$-clique in $G$ will have exactly one edge from $E_{\{i,j\}}$ for each choice of $i,j$.

  The structure of our \PBS{} instance will force us to choose a vertex $v_i$ from each class $V_i$, corresponding to the vertices of a $k$-clique.
  In addition, for each chosen vertex $v_i$ and each $j \neq i$, we choose an edge $e_{i\rightarrow j}$ between $v_i$ and $V_j$.
  A set of $O(k^2)$ Checksum gadgets will ensure that for each $i \neq j$, the chosen edges $e_{i \rightarrow j}$ and $e_{j \rightarrow i}$ must be the same.
  This ensures that $v_i$ and $v_j$ are adjacent for each $i \neq j$, and that therefore the vertices $v_1, \dots, v_k$ form a clique.

  We build our \PBS{} instance $(D,w,X)$ out of Duplication, Choice and Checksum gadgets, as follows.
  Let {\sc Start} be a Duplication gadget with input arc $a^*$, and $k$ output arcs. Label each output arc with a different integer $i$ from $[k]$.
    For each $i \in [k]$, let {\sc ChooseVertex}$(i)$ be a Choice gadget with $|V_i|$ output arcs. Label each output arc with a different vertex $v$ from $V_i$.
  Join the input arc of {\sc ChooseVertex}$(i)$ to the output arc of {\sc Start} with label $i$.
  
  For each $i \in [k], v \in V_i$, let
  {\sc AssignVertex}$(i,v)$ be a Duplication gadget with $k-1$ output arcs. 
  Label the output arcs with the integers from $[k]\setminus\{i\}$.
  Join the input arc of {\sc AssignVertex}$(i,v)$ to the output arc of {\sc ChooseVertex}$(i)$ with label $v$.

  For each $i \in [k], v \in V_i, j \in [k] \setminus \{i\}$ let {\sc ChooseEdge}$(i,v, \rightarrow j)$ be a Choice gadget with $|N(v) \cap V_j|$ output arcs. 
  Label each output arc with a different edge $e_r$ between $v$ and $V_j$.
  Join the input arc of {\sc ChooseEdge}$(i,v, \rightarrow j)$ to the output arc of {\sc AssignVertex}$(i,v)$ with label $j$.
  
  For each $i \in [k], v \in V_i, j \in [k] \setminus \{i\}$ and edge $e_r$ between $v_i$ and $V_j$, let {\sc AssignEdge}$(i,v,\rightarrow j, e_r)$ be a Duplication gadget with $r$ output arcs.
  Label this whole set of output arcs as {\sc Output}$(i,v, \rightarrow j, e_r)$.
  Join the input arc of {\sc AssignEdge}$(i,v,\rightarrow j, e_r)$ to the output arc of {\sc ChooseEdge}$(i,v, \rightarrow j)$ with label $e_r$.
  
  Finally, for each $i,j \in [k]$ with $i < j$, let {\sc CheckEdge}$(i,j)$ be a Checksum gadget with 
  $\sum\{r: e_r \in E_{\{i,j\}}\}$ left input arcs and $\sum\{r: e_r \in E_{\{i,j\}}\}$ right input arcs.
  Partition the left and right input arcs of {\sc CheckEdge}$(i,j)$ as follows.
  For each $e_r \in E_{\{i,j\}}$, let {\sc Input}$(i,v, \rightarrow j, e_r)$ be a set of $r$ left input arcs, where $v$ is the endpoint of $e_r$ in $V_i$.
  Similarly, let {\sc Input}$(j,u, \rightarrow i, e_r)$ be a set of $r$ right input arcs, where $u$ is the endpoint of $e_r$ in $V_j$.
  Now, join each set of arcs of the form {\sc Input}$(i,v, \rightarrow j, e_r)$ to the set of arcs of the form {\sc Output}$(i,v, \rightarrow j, e_r)$ from the gadget {\sc AssignEdge}$(i,v, \rightarrow j, e_r)$.

  Finally, we assign weights. 
  Let $a^*$ have weight $-1$ and let all other arcs have weight $0$.
  This concludes the construction of $D$.
  Observe that every output arc is joined to an input arc, and every input arc except $a^*$ is joined to an output arc.
  
\medskip
    {\bf Correctness.}

  We now show that $D$ has a properly balanced subgraph of negative weight if and only if $G$ has a clique with $k$ vertices.
   Observe that by repeated use of Lemma \ref{lem:joinCorrectness}, a subgraph $D'$ of $D$ is a properly balanced subgraph if and only if 
  
  \begin{itemize}
   \item  $D'$ restricted to any gadget {\sc Start}, {\sc ChooseVertex}$(i)$, {\sc AssignVertex}$(i,v)$, {\sc ChooseEdge}$(i,v, \rightarrow j)$, {\sc AssignEdge}$(i,v,\rightarrow j, e_r)$ or {\sc CheckEdge}$(i,j)$ is a properly balanced subgraph; and
   \item for each output arc $a$ that is joined to an input arc $a'$, $a$ is in $D'$ if and only if $a'$ is in $D'$.
  \end{itemize}

  First suppose $G$ has a clique on $k$ vertices. 
  By definition of $G$, this clique must have exactly one vertex from each class $V_i, i \in [k]$.
  For each $i \in [k]$, let $v_i$ be the vertex of $V_i$ that is in the clique.
  For each $i \neq j$, let $r(i,j)$ be the index such that $e_{r(i,j)}$ is the edge between $v_i$ and $v_j$.
  
  We will now describe a graph $D'$ by describing its restriction to each gadget. The construction will be such that an output arc is in $D'$ if and only if the input arc it is joined to is also in $D'$.
  Refer to a gadget as \emph{passive} if no arcs in it are contained in $D'$, and \emph{active} otherwise; further, for a Choice gadget, say that it \emph{selects} arc $i$ of the $i$'th output arc is contained in $D'$. Note by previous Propositions that these options all correspond to restrictions of balanced subgraphs to $D'$. 

  The graph $D'$ is constructed as follows. The {\sc Start} gadget is active; 
  for every $i \in [k]$, the {\sc ChooseVertex}$(i)$ gadget is active and selects the output arc labelled $v_i$,
  and the {\sc AssignVertex}$(i,v)$ gadget is active for $v=v_i$; 
  and for every $i \in [k]$ and every $j \in [k] \setminus \{i\}$, the gadget {\sc ChooseEdge}$(i,v, \rightarrow j)$ with $v=v_i$ is active, selecting the output arc labelled $e_{r(i,j)}$,
  and the {\sc AssignEdge}$(i,v,\rightarrow j, e)$ gadget is active for $v=v_i$ and $e=e_{r(i,j)}$. All other {\sc AssignVertex}, {\sc ChooseEdge} and {\sc AssignEdge} gadgets are passive.
  Note that $D'$ contains an arc set {\sc Output}$(i,v,\rightarrow j, e_r)$ if and only if $v=v_i$ and $r=r(i,j)$.

  Finally, for each $i,j \in [k]$ with $i < j$, let $D'$ restricted to {\sc CheckEdge}$(i,j)$ be a properly balanced subgraph containing the left input arcs from {\sc Input}$(i,v_i, \rightarrow j, e_{r(i,j)})$, the right input arcs from {\sc Input}$(j,v_j, \rightarrow i, e_{r(j,i)})$, and no other input arcs. As $r(i,j)=r(j,i)$, such a subgraph exists by Proposition \ref{prop:gadgets}.
  
  This concludes the construction of $D'$. As $D'$ restricted to each gadget is a properly balanced subgraph, and an output arc is in $D'$ if and only if the input arc it is joined to is in $D'$, we have that $D'$ is a properly balanced subgraph of $D$. 
  As $D'$ contains the arc $a^*$ of weight $-1$ and all other arcs have weight $0$, $D'$ is a properly balanced subgraph with negative weight, as required.

  Now for the converse, suppose that $D$ has a properly balanced subgraph $D'$ of negative weight.
  Then $D'$ must contain $a^*$, the input arc of {\sc Start} with weight $-1$.
  By Proposition \ref{prop:gadgets}, $D'$ must contain all of the output arcs of {\sc Start}.
  Thus for each $i \in [k]$, $D'$ contains the input arc of {\sc ChooseVertex}$(i)$.
  By Proposition \ref{prop:gadgets}, $D'$ contains exactly one output arc of  {\sc ChooseVertex}$(i)$;
  let $v_i \in V_i$ be the unique vertex in $G$ such that $D'$ restricted to  {\sc ChooseVertex}$(i)$ contains the output arc labelled $v_i$.
  
  It now follows that for each $i \in [k], v \in V_i$, $D'$ contains the input arc of  {\sc AssignVertex}$(i,v)$ if and only if $v = v_i$. 
  Then by Proposition \ref{prop:gadgets}, if $v = v_i$ then $D'$ contains the all the output arcs of  {\sc AssignVertex}$(i,v)$, and otherwise $D'$ contains none of the output arcs of  {\sc AssignVertex}$(i,v)$.

  It follows that for each  $i \in [k], v \in V_i, j \in [k] \setminus \{i\}$, $D'$ contains the input arc of {\sc ChooseEdge}$(i,v, \rightarrow j)$ if and only if $v = v_i$.
  If $v \neq v_i$ then by Proposition \ref{prop:gadgets} $D'$ contains none of the output arcs of  {\sc ChooseEdge}$(i,v, \rightarrow j)$.
  If $v = v_i$, then again by Proposition \ref{prop:gadgets} $D'$ contains exactly one output arc of  {\sc ChooseEdge}$(i,v, \rightarrow j)$.
  So for each $i \in [k], j \in [k] \setminus \{i\}$, let $r(i \rightarrow j)$ be the index such that $D'$ contains the output arc of  {\sc ChooseEdge}$(i,v_i, \rightarrow j)$ labelled with $e_{r(i \rightarrow j)}$.
  (Later we will show that $r(i\rightarrow j) = r(j \rightarrow i)$, implying that $v_i$ and $v_j$ are adjacent.)
  
  It now follows that for each $i \in [k], v \in V_i, j \in [k] \setminus \{i\}$ and edge $e_r$ between $v_i$ and $V_j$, $D'$ contains the input arc of {\sc AssignEdge}$(i,v,\rightarrow j, e_r)$ if and only $v = v_i$ and $r = r(i \rightarrow j)$.
  Furthermore by Proposition  \ref{prop:gadgets}, $D'$ contains the set of output arcs {\sc Output}$(i,v,\rightarrow j, e_r)$ if $v = v_i$ and $r = r(i \rightarrow j)$, and otherwise $D'$ contains none of the arcs from {\sc Output}$(i,v,\rightarrow j, e_r)$.
  
  We now have that for each $i,j \in [k]$ with $i < j$, the left input arcs of {\sc CheckEdge}$(i,j)$ in $D'$ are exactly those in  {\sc Input}$(i,v_i,\rightarrow j, e_{r(i \rightarrow j)})$, and the right input arcs of {\sc CheckEdge}$(i,j)$ in $D'$ are exactly those in  {\sc Input}$(j,v_j,\rightarrow i, e_{r(j \rightarrow i)})$.
  By Proposition  \ref{prop:gadgets}, we have that $|\textsc{ Input}(i,v_i,\rightarrow j, e_{r(i \rightarrow j)})| = |\textsc{ Input}(j,v_j,\rightarrow i, e_{r(j \rightarrow i)})|$ and so $r(i \rightarrow j) = r(j \rightarrow i)$.
  It follows that $e_{r(i \rightarrow j)}$ and $e_{r(j \rightarrow i)}$ are the same edge, and that therefore this is an edge in $G$ between $v_i$ and $v_j$.
  
  Thus we have that $v_1, \dots, v_k$ form a clique in $G$, as required.

\medskip
  {\bf Structure of the constructed graph and wrap-up of proof}.

  Having showed that $D$ represents the instance of {\sc $k$-Multicolored Clique}, it remains to show that $D$ satisfies the specified properties, that $\pw(D)$ is bounded by a function of $k$, and that it can be constructed in fixed-parameter time. 
%

 We now address the properties of $D$ in turn. 
    It is clear that there exists a single arc $a^*$ of weight $-1$, that $a^*$ is not part of a double arc, that all other arcs have weight $0$, and that there are no forbidden pairs. 
    To see that pathwidth is bounded,
    let $D^*$ be the graph derived from $D$ by deleting the vertices $w$ and $z$ from every {\sc CheckEdge} gadget, along with incident arcs. 
    (That is, $D^*$ is the graph we had before adding {\sc CheckEdge} gadgets in the construction of $D$.)
    We constructed $D^*$ by joining arcs in {\sc Start} to the input arcs of the  {\sc ChooseVertex}$(i)$ gadgets, then joining arcs of the resulting graph to the input arcs of the  {\sc AssignVertex}$(i,v)$ gadgets, then joining arcs of the resulting graph to the input arcs of the {\sc ChooseEdge}$(i,v, \rightarrow j)$ gadgets, then joining arcs of the resulting graph to the input arcs of the {\sc AssignEdge}$(i,v,\rightarrow j, e_r)$ gadgets.  Observe that a Duplication gadget can be turned into a path by the removal of one vertex, and Choice and Checksum gadgets can be turned into disjoint unions of paths by the removal of two vertices. Therefore Duplication gadgets have pathwidth $2$, and Choice and Checksum gadgets have pathwidth $3$. 
    Hence by repeated use of Lemma \ref{lem:joinPathwidth}, $D^*$ has pathwidth at most $((((2 + 3 + 1) + 2 + 1) + 3 + 1) + 2 + 1) = 16$.
    
    There are $\binom{k}{2} = \frac{k^2-k}{2}$ {\sc CheckEdge} gadgets, and therefore we can remove $k^2-k$ vertices from $D$ to get $D^*$.
  It follows that $D$ has pathwidth at most $k^2 - k +16$ (as we can add the $k^2-k$ extra vertices to every bag in a path decomposition of $D^*$).

  Finally, it is clear that the reduction can be performed in polynomial time, as we construct a polynomial number of gadgets and each gadget can be constructed in polynomial time. Thus we have provided a fixed-parameter time reduction from any instance of {\sc $k$-Multicolor Clique} to an instance of \PBS{} with the required properties and with pathwidth $O(k^2)$. This concludes the proof.
  
 \qed
 \end{proof}

  \section{Omitted proofs and figures from Section \ref{sec3}}
 

 
 We first describe the gadgets of Section \ref{sec3}, and their solutions, in more detail.
 Let ($D = (V,A), w, X=\{(a_i,a_i'): i \in [t]\}, Y=\emptyset$) be an instance of \PBS{} with double arcs $X$ and no forbidden pairs, and where $w(a^*)=-1$ for a single arc $a^*$ and $w(a)=0$ for every other arc.
  For each pair of vertices $u,v$ such that there is an arc from $u$ to $v$ in $D$, we produce a gadget {\sc Gadget($u,v$)} that contains $u,v$ and new vertices appearing in {\sc Gadget($u,v$)}
 Recall that we may assume that between any pair of vertices, there is either an arc of weight $-1$, an arc of weight $0$, or a double arc with both arcs having weight $0$.

 Given a directed multigraph $H$ (corresponding to part of a solution to an \MCPP{} instance) and a vertex $v$, the \emph{imbalance} of $v$ is $d_H^+(v) - d_H^-(v)$.


  \begin{figure}[h]
  \centering
  \begin{minipage}[b]{0.24\linewidth}
    \centering
    \begin{tikzpicture}[>= triangle 60,scale=0.7]
      \node at (0,0) (U)[city] [label = above: $u$]{}; \node at (3,0)
      (V)[city] [label = above: $v$]{} edge[<-] node[above] {$0$}(U);
      \node at (0,-2.4){};
    \end{tikzpicture}
    \subcaption{Arc $uv$ of weight $0$}
    \label{subfig:zeroArcOriginal}
  \end{minipage}%
  \begin{minipage}[b]{0.24\linewidth}
    \centering
    \begin{tikzpicture}[>= triangle 60,scale=0.7]
      \node at (0,0) (U)[city] [label = above right: $u$]{}; \node at
      (2,0) (Zuv)[city] [label = above right:$z_{uv}$ ]{} edge[->](U);
      \node at (4,0) (V)[city] [label = above right: $v$]{}
      edge[<-](Zuv); \draw[-latex, ->, dashed] (U)
      to[out=90,in=100,looseness=3] (Zuv); \draw[-latex, ->, dashed] (U)
      to[out=270,in=260,looseness=3] (Zuv); \draw[-latex, ->, dashed]
      (V) to[out=90,in=80,looseness=3] (Zuv); \node at (0,-2.4){};
    \end{tikzpicture}
    \subcaption{{\sc Gadget($u,v$)}}
    \label{subfig:zeroArcGadget}
  \end{minipage}%
  \begin{minipage}[b]{0.24\linewidth}
    \centering
    \begin{tikzpicture}[>= triangle 60,scale=0.7]
      \node at (0,0) (U)[city] [label = right: $u$]{};
      \node at (2,0) (Zuv)[city] [label = above right:$z_{uv}$ ]{}
      edge[->, bend left = 50](U)
      edge[->, bend right = 50](U);
      \node at (4,0) (V)[city] [label = above right: $v$]{}
      edge[<-](Zuv);
      \draw[-latex, ->, dashed] (U) to[out=90,in=100,looseness=3] (Zuv);
      \draw[-latex, ->, dashed] (U) to[out=270,in=260,looseness=3] (Zuv);
      \draw[-latex, ->, dashed] (V) to[out=90,in=80,looseness=3] (Zuv);
      \node at (0,-2.4){};
    \end{tikzpicture}
    \subcaption{Passive solution}
    \label{subfig:zeroArcPassive}
  \end{minipage}%
  \begin{minipage}[b]{0.24\linewidth}
    \centering
    \begin{tikzpicture}[>= triangle 60,scale=0.7]
      \node at (0,0) (U)[city] [label = above right: $u$]{};
      \node at (2,0) (Zuv)[city] [label = above right:$z_{uv}$ ]{}
      edge[->](U);
      \node at (4,0) (V)[city] [label = above right: $v$]{}
      edge[<-, bend left = 50](Zuv)
      edge[<-, bend right = 50](Zuv);
      \draw[-latex, ->, dashed] (U) to[out=90,in=100,looseness=3] (Zuv);
      \draw[-latex, ->, dashed] (U) to[out=270,in=260,looseness=3] (Zuv);
      \draw[-latex, ->, dashed] (V) to[out=90,in=80,looseness=3] (Zuv);
      \node at (0,-2.4){};
    \end{tikzpicture}
    \subcaption{Active solution}
    \label{subfig:zeroArcActive}
  \end{minipage}    
  \caption{The gadget {\sc Gadget($u,v$)} when $uv$ is an arc of weight $0$ that is not part of a double arc. Dashed lines represent heavy arcs.}\label{fig:zeroArcGadget}
\end{figure}
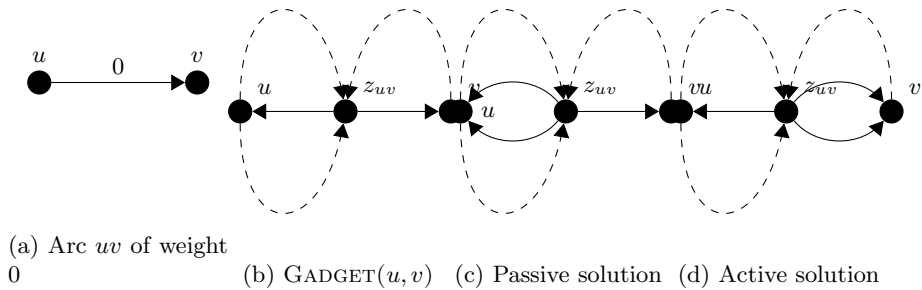

 \medskip
 {\bf For an arc $uv$ of weight $0$ that is not part of a double arc:} Construct {\sc Gadget($u,v$)} by creating
%
%
%
a new vertex $z_{uv}$, with standard arcs $z_{uv}u$ and $z_{uv}v$, two heavy arcs $uz_{uv}$, and a heavy arc $vz_{uv}$. 
(See Figure \ref{fig:zeroArcGadget}.)

 For any solution in which each heavy arc is traversed exactly once, $z_{uv}$ has in-degree exactly $3$. The remaining arcs are the two out-arcs of $z_{uv}$, which must be traversed exactly three times between them. There are therefore two possible choices:
 
 \begin{itemize}
  \item {\bf Passive Solution:} Traverse $z_{uv}u$ twice and $z_{uv}v$ once. In this solution, every vertex is balanced within {\sc Gadget($u,v$)} and the cost is $3M+3$. 
  \item {\bf Active Solution:} Traverse $z_{uv}u$ once and $z_{uv}v$ twice. In this solution, every vertex is balanced except for $u$, which has imbalance $1$, and $v$, which has imbalance $-1$. The cost of this solution is also $3M+3$. 
 \end{itemize}

 Observe that the difference between the weight of the passive and active solutions is $0$, and the imbalance at $u$ and $v$ for the active solution is the same as in an arc from $u$ to $v$.

 The total weight of {\sc Gadget($u,v$)} is $3M+2$. 

 \begin{figure}[h]
  \centering
  \begin{minipage}[t]{0.48\linewidth}
    \centering
    \begin{tikzpicture}[>= triangle 60]
      \node at (0,0) (U)[city] [label = above: $u$]{};
      \node at (2,0) (V)[city] [label = above: $v$]{}
      edge[<-] node[above] {$-1$}(U);
      \node at (0,-2.4){};
    \end{tikzpicture}
    \subcaption{Arc $uv$ of weight $-1$}
    \label{subfig:negativeArcOriginal}
  \end{minipage}%
  \begin{minipage}[t]{0.48\linewidth}
    \centering
    \begin{tikzpicture}[>= triangle 60]
      \node at (0,0) (U)[city] [label = above right: $u$]{};
      \node at (2,0) (Zuv)[city] [label = above right:$z_{uv}$ ]{}
      edge[->](U);
      \node at (4,0) (Wuv)[city] [label = above right:$w_{uv}$ ]{}
      edge[<-](Zuv);
      \node at (6,0) (V)[city] [label = above right: $v$]{}
      edge[->](Wuv);
      \draw[-latex, ->, dashed] (U) to[out=90,in=100,looseness=3] (Zuv);
      \draw[-latex, ->, dashed] (U) to[out=270,in=260,looseness=3] (Zuv);
      \draw[-latex, ->, dashed] (Wuv) to[out=90,in=80,looseness=3] (Zuv);
      \draw[-latex, ->, dashed] (Wuv) to[out=90,in=100,looseness=3] (V);
      \draw[-latex, ->, dashed] (Wuv) to[out=270,in=260,looseness=3] (V);
      \node at (0,-2.4){};
    \end{tikzpicture}
    \subcaption{{\sc Gadget($u,v$)}}
    \label{subfig:negativeArcGadget}
  \end{minipage}
  ~
  \begin{minipage}[t]{0.48\linewidth}
    \centering
    \begin{tikzpicture}[>= triangle 60]
      \node at (0,0) (U)[city] [label = right: $u$]{};
      \node at (1.5,0) (Zuv)[city] [label = above right:$z_{uv}$ ]{}
      edge[->, bend left = 50](U)
      edge[->, bend right=50](U);
      \node at (3,0) (Wuv)[city] [label = right:$w_{uv}$ ]{}
      edge[<-](Zuv);
      \node at (4.5,0) (V)[city] [label = above right: $v$]{}
      edge[->, bend right=50](Wuv)
      edge[->, bend left=50](Wuv);
      \draw[-latex, ->, dashed] (U) to[out=90,in=100,looseness=3] (Zuv);
      \draw[-latex, ->, dashed] (U) to[out=270,in=260,looseness=3] (Zuv);
      \draw[-latex, ->, dashed] (Wuv) to[out=90,in=80,looseness=3] (Zuv);
      \draw[-latex, ->, dashed] (Wuv) to[out=90,in=100,looseness=3] (V);
      \draw[-latex, ->, dashed] (Wuv) to[out=270,in=260,looseness=3] (V);
      \node at (0,-2.4){};
    \end{tikzpicture}
    \subcaption{Passive solution}
    \label{subfig:negativeArcPassive}
  \end{minipage}%
  \begin{minipage}[t]{0.48\linewidth} 
    \centering
    \begin{tikzpicture}[>= triangle 60]
      \node at (0,0) (U)[city] [label = above right: $u$]{};
      \node at (1.5,0) (Zuv)[city] [label = right:$z_{uv}$ ]{}
      edge[->](U);
      \node at (3,0) (Wuv)[city] [label = above right:$w_{uv}$ ]{}
      edge[<-, bend left=50](Zuv)
      edge[<-, bend right=50](Zuv);
      \node at (4.5,0) (V)[city] [label = above right: $v$]{}
      edge[->](Wuv);
      \draw[-latex, ->, dashed] (U) to[out=90,in=100,looseness=3] (Zuv);
      \draw[-latex, ->, dashed] (U) to[out=270,in=260,looseness=3] (Zuv);
      \draw[-latex, ->, dashed] (Wuv) to[out=90,in=80,looseness=3] (Zuv);
      \draw[-latex, ->, dashed] (Wuv) to[out=90,in=100,looseness=3] (V);
      \draw[-latex, ->, dashed] (Wuv) to[out=270,in=260,looseness=3] (V);
      \node at (0,-2.4){};
    \end{tikzpicture}
    \subcaption{Active solution}
    \label{subfig:negativeArcActive}
  \end{minipage}
  \caption{The gadget {\sc Gadget($u,v$)} when $uv$ is an arc of weight $-1$ that is not part of a double arc. Dashed lines represent heavy arcs.}\label{fig:negativeArcGadget}
\end{figure}
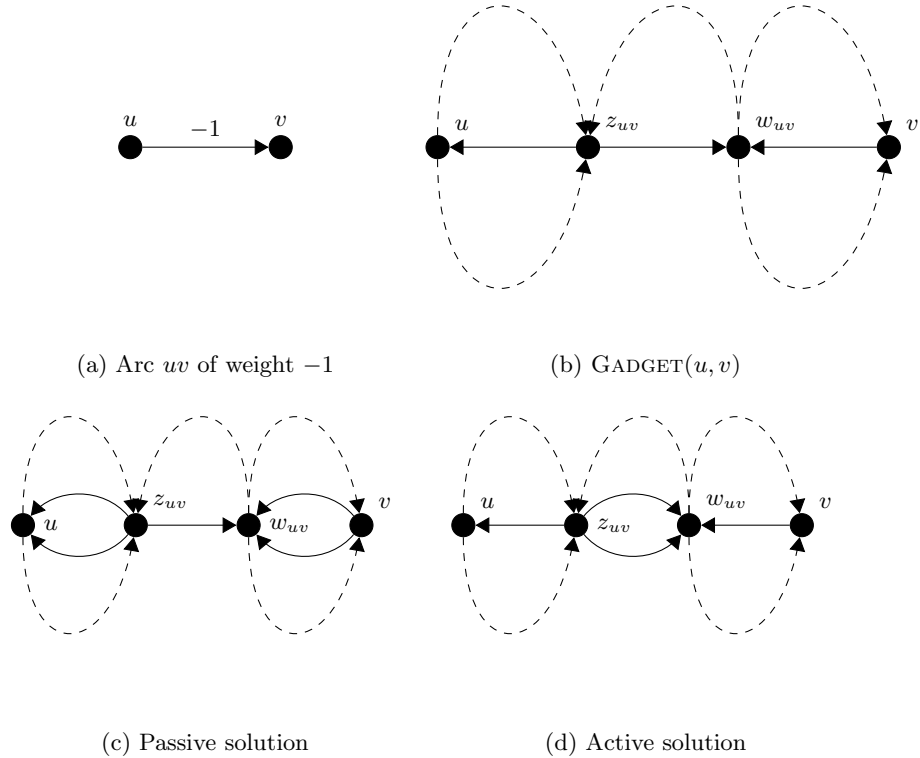

\medskip

 {\bf For an arc $uv$ of weight $-1$ that is not part of a double arc:}
%
  Construct {\sc Gadget($u,v$)} by adding
  two new vertices $w_{uv}$ and $z_{uv}$, with standard arcs $z_{uv}u, z_{uv}w_{uv}$ and $vw_{uv}$, two heavy arcs $uz_{uv}$, one heavy arc $w_{uv}z_{uv}$, and two heavy arcs $w_{uv}v$. 
(See Figure \ref{fig:negativeArcGadget}.)
 
 For any solution in which each heavy arc is traversed exactly once,  $z_{uv}$ will have exactly $3$ in-arcs, and $w_{uv}$ will have exactly $3$ out-arcs. It remains to decide how many times to use the arcs out of $z_{uv}$ and into $w_{uv}$. 
 There are two possible solutions:
 
  \begin{itemize}
  \item {\bf Passive Solution:} Traverse $z_{uv}u$ twice, $z_{uv}w_{uv}$ once and $vw_{uv}$ twice. In this solution, every vertex is balanced within {\sc Gadget($u,v$)} and the cost is $5M+5$. 
  \item  {\bf Active Solution:} Traverse $z_{uv}$ once, $z_{uv}w_{uv}$ twice and $vw_{uv}$ once. In this solution, every vertex is balanced except for $u$, which has imbalance $1$, and $v$, which has imbalance $-1$. The cost of this solution is $5M+4$. 
 \end{itemize}

 Observe that the active solution costs $1$ less than the passive solution, and the imbalance at $u$ and $v$ for the active solution is again the same as in an arc from $u$ to $v$.
 
 The weight of this gadget is $5M+3$.
 
 \medskip
 
  \begin{figure}[h]
  \centering
  \begin{minipage}[t]{0.27\linewidth}
    \centering
    \begin{tikzpicture}[>= triangle 60]
      \node at (0,0) (U)[city] [label = above: $u$]{};
      \node at (2,0) (V)[city] [label = above: $v$]{}
      edge[<-, bend left = 50] node[below] {$0$}(U)
      edge[<-, bend right = 50] node[above] {$0$}(U);
    \end{tikzpicture}
    \subcaption{A double arc from $u$ to $v$}
    \label{subfig:doubleArcOriginal}
  \end{minipage}%
  \begin{minipage}[t]{0.23\linewidth}
    \centering
    \begin{tikzpicture}[>= triangle 60]
      \node at (0,0) (U)[city] [label = above: $u$]{};
      \node at (2,0) (V)[city] [label = above: $v$]{}
      edge[dotted, bend right = 50](U)
      edge[<-, dashed, bend left = 50] (U);
    \end{tikzpicture}
    \subcaption{{\sc Gadget($u,v$)}}
    \label{subfig:doubleArcGadget}
  \end{minipage}%
  \begin{minipage}[t]{0.23\linewidth}
    \centering
    \begin{tikzpicture}[>= triangle 60]
      \node at (0,0) (U)[city] [label = above: $u$]{};
      \node at (2,0) (V)[city] [label = above: $v$]{}
      edge[->, dashed, bend right = 50] (U)
      edge[<-, dashed, bend left = 50] (U);
    \end{tikzpicture}
    \subcaption{Passive solution}
    \label{subfig:doubleArcPassive}
  \end{minipage}%
  \begin{minipage}[t]{0.23\linewidth}
    \centering  
    \begin{tikzpicture}[>= triangle 60]
      \node at (0,0) (U)[city] [label = above: $u$]{};
      \node at (2,0) (V)[city] [label = above: $v$]{}
      edge[<-, dashed, bend right = 50] (U)
      edge[<-, dashed, bend left = 50] (U);
    \end{tikzpicture}
    \subcaption{Active solution}
    \label{subfig:doubleArcActive}
  \end{minipage}
  \caption{The gadget {\sc Gadget($u,v$)} when there is a double arc from $u$ to $v$. Dashed lines represent heavy arcs; the dotted line represents a heavy edge.}\label{fig:doubleArcGadget}
\end{figure}
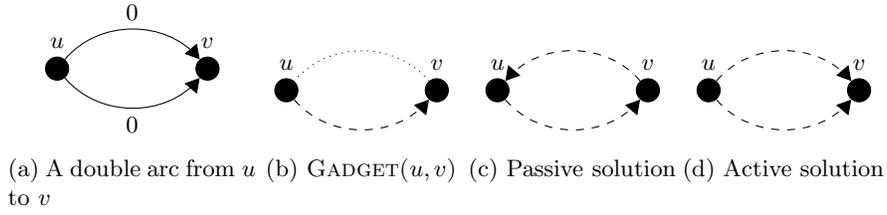
 
 {\bf For a double arc from $u$ to $v$:}
 {\sc Gadget($u,v$)} consists of a heavy arc $uv$ and a heavy edge $\{u,v\}$. 
(See Figure \ref{fig:doubleArcGadget}.)
 Assuming a solution traverses each heavy arc/edge exactly once, the only thing to decide is in which direction to traverse the undirected edge.
 Thus there are two possible solutions:
 
   \begin{itemize}
  \item {\bf Passive Solution:} Traverse the edge from $v$ to $u$. In this solution, every vertex is balanced within {\sc Gadget($u,v$)} and the cost is $2M$. 
  \item {\bf Active Solution:} Traverse the edge from $u$ to $v$. In this solution, every vertex is balanced except for $u$, which has imbalance $2$, and $v$, which has imbalance $-2$. The cost of this solution is also $2M$. 
 \end{itemize}

 Observe that the difference between the weight of the passive and active solutions is $0$, and the imbalance at $u$ and $v$ for the active solution is the same as in a double arc from $u$ to $v$.

 The weight of this gadget is $2M$.
 
 
\medskip

{\bf Removing heavy arcs and edges.} Finally, we replace every heavy arc (edge, respectively) we just created, of weight $M$, by a directed (undirected) path of length $M$, where all internal vertices have degree two. 
Note that in any minimal solution, each arc or edge in such a path will be traversed the same number of times
(if one edge in the path is traversed more times than its neighbor, then it must be traversed at least twice more, including at least once in each direction, and so the solution is not minimal).
Thus in the analysis, we may treat such a path as effectively being a single edge or arc of weight $M$.
 
 \medskip
 
 In the proof of Theorem \ref{thm:mcppTreewidthHardness} we will us the following technical lemma, whose proof is identical to that of Lemma \ref{lem:joinPathwidth} (see Section \ref{appendixA}).

 \begin{lemma}\label{lem:gadgetPathwidth}
  Let $H$ be a mixed multigraph, and $G$ the mixed multigraph derived by replacing each arc or double arc from $u$ to $v$ with a gadget $G_{uv}$. 
  Then $\pw(G) \le \pw(H) + \max_{uv} \pw(G_{uv})$.
 \end{lemma}
 

We now give the full proof of Theorem \ref{thm:mcppTreewidthHardness}.

\medskip
 {\bf \noindent Theorem \ref{thm:mcppTreewidthHardness}.}
{\em \MCPP{} is W[1]-hard parameterized by pathwidth.}
 
 \begin{proof}
 Let ($D = (V,A), w, X=\{(a_i,a_i'): i \in [t]\}, Y=\emptyset$) be an instance of \PBS{} with double arcs $X$ and no forbidden pairs, and where $w(a^*)=-1$ for a single arc $a^*$ and $w(a)=0$ for every other arc.
 Let $G$ be the graph in which every arc / double arc in $D$ is replaced with its corresponding gadget.
 
  Let $m_1$ be the number of arcs of weight 0 not in a double arc in $D$.
  By Theorem \ref{thm:mbsdaHardness}, we may assume there is only one arc of weight $-1$ not in a double arc in $D$.
 Let $m_2$ be the number of double arcs in $D$.
 
  If we use the passive solution for every gadget, then every vertex is balanced, every arc and edge is covered, and the total cost is 
 $m_1(3M+3) + 5M + 5 + m_2(2M) = (3m_1 + 2m_2 + 5)M + 3m_1 + 5$. Therefore this is an upper bound on the weight of a optimal solution.
 
 The total weight of the graph is $m_1(3M+2) + 5M + 3 + m_2(2M) = (3m_1 + 2m_2 + 5)M + 2m_1 + 3$. Therefore, any minimal solution that does not traverse each heavy arc/edge exactly once will have weight at least $(3m_1 + 2m_2 + 6)M + 2m_1 + 3$. This is $M - m_1 - 2$ greater than the solution in which we use the passive solution for every gadget. So by setting $M$ to be $m_1 + 3$, we may assume that the optimal solution traverses each heavy arc/edge exactly once. Therefore we may assume that we use either the active or passive solution for each gadget.
 
 
 Let $W = (3m_1 + 2m_2 + 5)M + 3m_1 + 5$, the cost of using the passive solution for each gadget.
 We now show that $G$ has a solution of weight less than $W$ if and only if $D$ has a solution of negative weight.

 Suppose first that $G$ has a solution of weight less than $W$. 
 As discussed above, we may assume that every gadget is either given the active or passive solution.
 Let $u^*,v^*$ be the vertices such that $u^*v^*$ is the only arc in $D$ of weight $-1$.
 Then {\sc Gadget($u^*,v^*$)} is the only gadget for which one solution weighs less than the other. 
 Therefore we may assume  {\sc Gadget($u^*,v^*$)} is given the active solution.
 
 
 Let $D'$ be the subgraph of $D$ whose arc set is the set of all arcs whose corresponding gadget in $G$ has the active solution (with both arcs in a double arc included if their gadget has the active solution, and neither if it has the passive solution).
 Then $D'$ contains $u^*v^*$ and so $D$ has negative weight. By construction, 
 $D'$ respects double arcs, and as there are no forbidden pairs $D'$ trivially respects forbidden pairs.
 It remains to show that $D'$ is balanced.

In our \MCPP{} solution, the imbalance of a vertex $v$ is equal to the sum of its imbalance in the active gadgets containing it. The imbalance of $v$ in an active gadget is $+1$ if the gadget corresponds to a single out-arc of $v$, $-1$ if the gadget corresponds to a single in-arc of $v$, $+2$ if the gadget corresponds to a double arc starting at $v$, and $-2$ if the gadget corresponds to a double arc ending at $v$. It follows that the imbalance of $v$ in our \MCPP{} solution is equal to its imbalance in $D'$,
Then as our \MCPP{} solution is balanced, $D'$ is also balanced, as required.
 
 Suppose on the other hand that $D$ has a solution $D'$ of negative weight. Construct a solution to \MCPP{} on $G$ by assigning each gadget the active assignment if the corresponding arc/double-arc appears in $D'$, and the passive assignment otherwise. As $D$ has negative weight, it must use $u^*v^*$ and so {\sc Gadget($u^*,v^*$)} gets the active solution. It follows that the cost of this solution is $W-1$. It is clear that every arc and edge is traversed at least once. As before, the imbalance of each vertex in this solution is equal to its imbalance in $D'$. Therefore this solution is balanced, and so corresponds to a closed walk, as required.

%
%
%
%
%
 
 
 We now show that $G$ has pathwidth bounded by a function of $\pw(D)$, the pathwidth of $D$. 

 

Observe that for each gadget {\sc Gadget($u,v$)} in our construction of $G$, {\sc Gadget($u,v$)} can be turned into a disjoint union of paths by the removal of at most $4$ vertices, and therefore {\sc Gadget($u,v$)} has pathwidth at most $5$.
Furthermore, $G$ can be derived from $D$ by replacing each arc or double arc with a corresponding gadgets. It follows from Lemma \ref{lem:gadgetPathwidth} that $G$ has pathwidth at most $\pw(D)+5$.


Let $m$ be the number of arcs in $D$. Then $G$ is derived from $D$ by introducing at most $m$ gadgets, and each gadget has at most $5M \le 5m + 15$ arcs. Therefore $G$ can be constructed in $O(m^2)$ time.

We now have that, given an instance $(D,w,X, Y)$  of \PBS{} of the type specified in Theorem \ref{thm:mbsdaHardness}, we can in polynomial time create an instance $G$ of \MCPP{} with pathwidth bounded by $\max(\pw(D),4)$. We therefore have a parameterized reduction from this restriction of \PBS, parameterized by pathwidth, to \MCPP{} parameterized by pathwidth. 
As this restriction of \PBS{} is W[1]-hard by Theorem \ref{thm:mbsdaHardness}, this completes the proof.
\qed
\end{proof}


%

\section{Omitted proofs from Section \ref{sec5}}

{\bf \noindent Lemma \ref{lem:naiveCPPBound}.}
{\em  Given an instance $(G,w)$ of \MCPP{} with $m$ arcs and edges, we can, in polynomial time, find closed walk of of $G$ that traverses each edge and arc at least once, if such a walk exists, and this walk traverses each arc at most $m+1$ times. 
}

  \begin{proof}
   For each arc or edge $a$ in $G$, we will construct a closed walk ${\textsc walk}_a$ that contains $a$.
   If $a$ is an edge, we may let ${\textsc walk}_a$ be a walk that traverses $a$ once in each direction. Then this is a closed walk that traverses $a$ twice and traverses no other arcs or edges.
   If $a$ is an arc,
   let $u$ be the start vertex and $v$ the end vertex of $a$. In polynomial time, find a simple walk from $v$ to $u$. (If no such walk exists, then $(G,w)$ has no solution). Adding $a=uv$ to this walk, we have a closed walk that traverse $a$ once and each other edge or arc at most once.
   Taking the union of walks ${\textsc walk}_a$ for every edge and arc $a$, we get a walk that traverses each arc or edge $a$ at most $m+1$ times (at most twice by ${\textsc walk}_a$, and at most once for all other walks).
   \qed
  \end{proof}

  \medskip
{\bf \noindent Lemma \ref{lem:changeGadgets}.}
{\em  Let $uv$ be an edge or arc in $G$, and let $B$ and $W$ be arbitrary integers such that $w(H_{uv})+W \leq M$. 
  Then the following are equivalent.
\begin{enumerate}
 \item There exists a graph $H'_{uv}$ with vertex set $\{u,v\}$ that covers $G_{uv}$, 
 such that  $w(H'_{uv}) = w(H_{uv})+W$ and $b_{H'_{uv}}(u) = b_{H_{uv}}(u) + B$;
 \item $D_{uv}$ has a subgraph $D'_{uv}$ which respects double arcs and forbidden pairs, such that $w(D'_{uv}) = W$ and $b_{D'_{uv}}(u) = B$.
\end{enumerate}
}

\begin{proof}
We consider each of the cases separately.

\medskip 
{\bf If $G_{uv}$ is an arc from $u$ to $v$ and $H_{uv}$ traverses $uv$  $t \le M$ times:} Then recall $D_{uv}$ has $t-1$ arcs from $v$ to $u$ of weight $-1$, and $M-t$ arcs from $u$ to $v$ of weight $1$.
Observe that  $t = w(H_{uv})$.

(\ref{condH}) $\rightarrow$ (\ref{condD}): Observe that 
$B = W$. As $H'_{uv}$ covers $G_{uv}$, we have $1 \le w(H_{uv}) + W \le M$, and therefore $1-t \le W \le M-t$.
If $W = 0$, let $D'_{uv}$ have no arcs. If $W > 0$, let $D'_{uv}$ have $W$ positive weight arcs from $u$ to $v$. If $W < 0$, let $D'_{uv}$ have $|W|$ negative weight arcs from $v$ to $u$.
Observe that in each case $D'_{uv}$ is a subgraph of $D_{uv}$ and satisfies condition \ref{condD}.

 (\ref{condD}) $\rightarrow$ (\ref{condH}): Observe that by construction of $D_{uv}$, $W=B$, and also observe that $1-t \le W$. If $W = 0$, then let $H'_{uv} = H_{uv}$ and observe that $H'_{uv}$ satisfies condition \ref{condH}.
 If $W > 0$, let $H'_{uv}$ be $H_{uv}$ with the addition of $W$ extra arcs from $u$ to $v$. Then $w(H'_{uv}) = w(H_{uv}) + W$, and $H'_{uv}$ satisfies condition \ref{condH}.
 If $W < 0$, let $H'_{uv}$ be $H_{uv}$ with $|W|$ arcs from $u$ to $v$ removed. As $|W| < t$, $H'_{uv}$ still has at least one arc from $u$ to $v$ and so covers $G_{uv}$, and condition \ref{condH} is satisfied.
 
 \medskip
 
 {\bf If $G_{uv}$ is an edge between $u$ and $v$, and $H_{uv}$ traverses $uv$  $t \le M$ times from $u$ to $v$, and from $v$ to $u$ $0$ times:}
Then recall that $D_{uv}$ has a double arc $(a,a')$, where $a$ and $a'$ are both arcs from $v$ to $u$ of weight $0$. In addition, $D_{uv}$ has $t-1$ arcs from $v$ to $u$ of weight $-1$, $M-t$ arcs from $u$ to $v$ of weight $1$, and $M-1$ arcs from $v$ to $u$ of weight $1$.

(\ref{condH}) $\rightarrow$ (\ref{condD}): Let $t'$ be the number of arcs from $u$ to $v$ in $H'_{uv}$, and $s'$ the number of arcs from $v$ to $u$ in $H'_{uv}$.
Observe that $W = s'+t'-t$ and $B = t'-t-s'$.

If $s' = 0$, then $t' > 0$ and we have $W = t'-t = B$. If $W = 0$, then let $D'_{uv}$ have no arcs. If $W > 0$, then note that $W \le M-t$ and let $D'_{uv}$ have $W$ positive weight arcs from $u$ to $v$. If $W < 0$, then note that $|W|\le t-1$ (as otherwise $H'_{uv}$ does not cover $uv$) and let $D'_{uv}$ have $|W|$ negative weight arcs from $v$ to $u$. Observe that in each case $D'_{uv}$ satisfies condition \ref{condD}.

If $t' = 0$, then $s' > 0$ and we have $W = s'-t$ and $B =  -t-s'$.
Then let $D'_{uv}$ have all $t-1$ negative weight arcs from $v$ to $u$, both arcs in the double arc from $v$ to $u$, and $s'-1$ positive weight arcs from $v$ to $u$.
Observe that $D'_{uv}$ has weight $s'-1 - (t-1)=W$ and $b_{D'_{uv}}(u) = -t-s' = B$, and $D'_{uv}$ respects double arcs, and so $D'_{uv}$ satisfies condition \ref{condD}.

If $t'>0$ and $s'>0$, then if $t'+s'>2$ we may remove a pair of arcs $(uv,vu)$ from $H'$ and get a better solution to \MCPP{}. Therefore we may assume $t'=s'= 1$ and so $W=2-t$ and $B = -t$. 
If $t\ge 2$ let $D'_{uv}$ contain $t-2$ negative weight arcs from $v$ to $u$, and both arcs of the double arc from $v$ to $u$. 
Otherwise $t = 1$. In this case, let $D'_{uv}$ contain both arcs of the double arc from $v$ to $u$, and one positive weight arc from $u$ to $v$.
In either case, $D'_{uv}$ has weight $-(t-2)=W$ and $b_{D'_{uv}}(u) = -t = B $, and $D'_{uv}$ respects double arcs, and so $D'_{uv}$ satisfies condition \ref{condD}.

(\ref{condD}) $\rightarrow$ (\ref{condH}): Let $s'$ be the number of positive weight arcs from $u$ to $v$ in $D'_{uv}$, $t'$ the number of negative weight arcs from $v$ to $u$ in $D'_{uv}$, and $r'$ the number of positive weight arcs from $v$ to $u$ in $D'_{uv}$. Then $W = s'-t'+r'$.

Suppose first that $D'_{uv}$ does not contain the double arc. Then $B = s'-t'-r'$.
If $s'-t' = 0$, then let $H'_{uv}$ be $H_{uv}$ with $r'$ arcs from $v$ to $u$ added.
If $s'-t' > 0$, then let $H'_{uv}$ be $H_{uv}$ with $s'-t'$ arcs from $u$ to $v$ added and $r'$ arcs from $v$ to $u$ added.
If $s' - t' <0$, then let $H'_{uv}$ be $H_{uv}$ with $t'-s'$ arcs from $u$ to $v$ removed and $r'$ arcs from $v$ to $u$ added (note that as $t' < t$, removing $t'-s'$ arcs from $u$ to $v$ still leaves $uv$ covered).
Observe that in each case, $H'_{uv}$ satisfies condition \ref{condH}.

Now suppose that  $D'_{uv}$ contains the double arc. Then $B = s'-t'-r'-2$.
If $s'-t' = 0$, then let $H'_{uv}$ be $H_{uv}$ with one arc from $u$ to $v$ removed and $r'+1$ arcs from $v$ to $u$ added. 
If $s'-t' > 0$, then let $H'_{uv}$ be $H_{uv}$ with $s'-t'-1$ arcs from $u$ to $v$ added and $r'+1$ arcs from $v$ to $u$ added.
It $s' - t' <0$, then let $H'_{uv}$ be $H_{uv}$ with $t'-s'+1$ arcs from $u$ to $v$ removed and $r'+1$ arcs from $v$ to $u$ added.
Observe that in each case, $H'_{uv}$ satisfies condition \ref{condH}.

\medskip

 {\bf If $G_{uv}$ is an edge between $u$ and $v$ and $H_{uv}$ traverses $uv$ $t>0$ times from $u$ to $v$, and $s>0$ times from $v$ to $u$:}
 Then recall that we may assume $s=t=1$, 
and $D_{uv}$ has a forbidden pair $(a,a')$, where $a$ is an arc from $u$ to $v$, $a'$ is an arc from $v$ to $u$, and both $a$ and $a'$ have weight $-1$. In addition, $D_{uv}$ has $M-1$ arcs from $u$ to $v$ of weight $1$, and $M-1$ arcs from $v$ to $u$ of weight $1$.

(\ref{condH}) $\rightarrow$ (\ref{condD}): If $H'_{uv}=H_{uv}$, then let $D'_{uv}$ have no arcs.
Otherwise, we may assume all arcs in $H'_{uv}$ are in the same direction (as otherwise we may remove a pair of arcs $\{uv,vu\}$ to get a better solution). So assume that all arcs in $H'_{uv}$ are from $u$ to $v$ (the other case is symmetric). 

Let $t'>0$ be the number of arcs from $u$ to $v$ in $H'_{uv}$. Then $W = t'-2$ and $B = t'$.
Then let $D'_{uv}$ contain the arc $uv$ from the forbidden pair, together with $t'-1$ positive weight arcs from $u$ to $v$.
Observe that $D'_{uv}$ satisfies condition \ref{condD}.

(\ref{condD}) $\rightarrow$ (\ref{condH}): 
Let $t'$ be the number of positive weight arcs from $u$ to $v$ in $D'_{uv}$, and let $s'$ be the number of positive weight arcs from $v$ to $u$ in $D'_{uv}$.

Suppose first that $D'_{uv}$ doesn't contain either arc from the forbidden pair. Then $W = t'+s'$ and $B = t'-s'$.
Then let $H'_{uv}$ be $H_{uv}$ with $t'$ arcs added from $u$ to $v$ and $s'$ arcs added from $v$ to $u$.
Observe that $H'_{uv}$ satisfies condition \ref{condH}.

So now suppose that $D'$ contains one arc from the forbidden pair; assume $D'$ contains the arc $uv$ (the other case is symmetric).
Then $W = t'+s'-1$ and $B = t'-s'+1$.
Then let $H'_{uv}$ be $H_{uv}$ with $t'$ arcs added from $u$ to $v$, $s'$ arcs added from $v$ to $u$ and one arc from $v$ to $u$ removed (note that even if $t'=s'=0$, removing the arc from $v$ to $u$ still leaves an arc from $u$ to $v$ covering $uv$).
Observe that $H'_{uv}$ satisfies condition \ref{condH}.
\qed
\end{proof}
 
 {\bf \noindent Lemma \ref{lem:MCPPtoPBSCorrectness}.}
{\em Let $D$ be the directed multigraph derived from $G$ and $H$ by taking the vertex set $V(G)$ and adding the gadget $D_{uv}$ for every arc and edge $uv$ in $G$.
 Then there exists a solution $H'$ with weight less than $H$ if and only if $D$ has a properly balanced subgraph of weight less than $0$.}
 
 \begin{proof}
 Suppose first that $H'$ is a solution with weight less than $w(H)$, and let $H'_{uv}$ be the subgraph of $H'$ induced by $\{u,v\}$, for every edge or arc $uv$ in $G$.
 For each edge and arc $uv$, let $D'_{uv}$ be the subgraph of $D_{uv}$ corresponding to $H'_{uv}$ in the claim (which exists as $w(H'_{uv}) \ge 1$ for each $uv$, which in turn implies $w(H'_{uv}) < M$ for each $uv$).
 Let $D'$ be the union of all such $D'_{uv}$. 
 As each $D'_{uv}$ respects double arcs and forbidden pairs, so does $D'$. 
 By Lemma \ref{lem:changeGadgets}, the weight of $D'$ is $\sum_{uv}w(D'_{uv}) = \sum_{uv}(w(H'_{uv}) - w(H_{uv})) = w(H') - w(H) < 0$ (where the sums are taken over all edges and arcs $uv$).
 Finally, for each vertex $u$, the imbalance of $D'$ at $u$ is $\sum_{a \in A(u)}b_{D'_a(u)} = \sum_{a \in A(u)}(b_{H'_a}(u) - b_{H_a}(u)) = \sum_{a \in A(u)}b_{H'_a}(u) - \sum_{a \in A(u)}b_{H_a}(u) = 0-0 = 0$ (where $A(u)$ is the set of all edges or arcs containing $u$). Thus, $D'$ is balanced. It follows that $D'$ is a properly balanced subgraph of $D$ of negative weight, as required.
 
 Conversely, suppose that $D'$ is a properly balanced subgraph of $D$ with weight less than $0$,  and let $D'_{uv}$ be the subgraph of $D'$ induced by $\{u,v\}$, for every edge or arc $uv$ in $G$.
  For each edge and arc $uv$, let $H'_{uv}$ be the graph corresponding to $D'_{uv}$ in the claim.
  (To see that this exists, observe that for arc or edge $a$ such that $w(D'_{a})<0$ it must be the case that $w(H_{a}) + w(D'_{a}) \le M$, and so $H'_{a}$ exists, which implies that $w(H_{a}) + w(D'_{a}) > 0$. Thus $w(H_{a}) + w(D'_{a}>0$ for all $a$, and as $\sum_a(w(H_{a}) + w(D'_{a})) = w(H) + w(D') < M$, this implies that $w(H_{a}) + w(D'_{a})< M$ for all $a$.)  
  Let $H'$ be the union of all such $H'_{uv}$.
  As each $H'_{uv}$ covers $uv$, $H'$ covers all edges and arcs of $G$.
   By Lemma \ref{lem:changeGadgets}, the weight of $H'$ is $\sum_{uv}w(H'_{uv}) = \sum_{uv}(w(H_{uv}) + w(D'_{uv})) = w(H) + w(D') < w(H)$ (where the sums are taken over all edges and arcs $uv$).
   Finally, for each vertex $u$, the imbalance of $H'$ at $u$ is $\sum_{a \in A(u)}b_{H'_a(u)} = \sum_{a \in A(u)}(b_{H_a}(u) + b_{D'_a}(u)) = \sum_{a \in A(u)}b_{H_a}(u) + \sum_{a \in A(u)}b_{D'_a}(u) = 0+0 = 0$ (where $A(u)$ is the set of all edges or arcs containing $u$). Thus, $H'$ is balanced. It follows that $H'$ is a solution with weight less than $w(H)$, as required.
  \qed
\end{proof}

\end{document}